\documentclass{article}
\newif\iffullversion
\fullversiontrue

\usepackage{fullpage}
\usepackage{cite}
\usepackage[cmex10]{amsmath}

\usepackage{amsfonts,amsthm}
\usepackage{algorithm}
\usepackage{algorithmic}

\usepackage[pdftex]{graphicx}
\DeclareGraphicsExtensions{.pdf,.jpeg,.png}

\usepackage{hyperref}
\usepackage{url} 
\usepackage{array}
\usepackage{newfloat}
\DeclareFloatingEnvironment[
    fileext=los,
    listname=List of Goals,
    name=Optimization Goal,
    placement=tbhp,
]{goal}
\usepackage{fixltx2e}
\usepackage[caption=false,font=footnotesize]{subfig}
\usepackage{pgfplots}  
\usepackage{tikz}
\usepackage{stfloats}

\pgfplotsset{compat=1.5} 

\newtheorem{theorem}{Theorem}

\newcommand{\PasswordSpace}{\mathcal{P}}
\newcommand{\cut}[1]{}
\def \QED {\hfill{$\Box$}}

\newenvironment{proofof}[1]{\noindent {\em Proof of #1.  }}{\QED}

\newenvironment{remindertheorem}[1]{\medskip \noindent {\bf Reminder of Theorem #1.  }\em}{}

\newcommand{\Hash}{\mathbf{H}}
\newcommand{\E}{\mathbb{E}}
\newcommand{\Cost}[1]{\mathbf{Cost}\left(#1 \right)}
\newcommand{\CostH}{\Cost{\Hash}}
\newcommand{\CostHk}{\Cost{\Hash^k}}
\newcommand{\PAdvDet}[2]{\mathcal{P}_{ADV,#1,#2}^{det}}
\newcommand{\PAdvUnif}[2]{\mathcal{P}_{ADV,#1,#2}^{pepper}}
\newcommand{\PAdvCASH}[3]{\mathcal{P}_{ADV,#1,#2,#3}^{CASH}}
\newcommand{\UAdvDet}[3]{\mathbf{U}_{ADV}^{det}\left(#1,#2,#3 \right)}
\newcommand{\UAdvCASH}[2]{\mathbf{U}_{ADV}^{CASH}\left(#1,#2 \right)}
\newcommand{\ServerCost}[1]{C_{SRV,#1}}
\newcommand{\FullVersion}[2]{#2}
\begin{document}
\raggedbottom

\title{ CASH: A Cost Asymmetric Secure Hash Algorithm for Optimal Password Protection}


\author{Jeremiah Blocki \\
Microsoft Research 
\and
Anupam Datta\\ Carnegie Mellon University
} 


\maketitle
\begin{abstract}
An adversary who has obtained the cryptographic hash of a user's password can mount an offline attack to crack the password by comparing this hash value with the cryptographic hashes of likely password guesses. This offline attacker is limited only by the resources he is willing to invest to crack the password. Key-stretching techniques like hash iteration and memory hard functions have been proposed to mitigate the threat of offline attacks by making each password guess more expensive for the adversary to verify. However, these techniques also increase costs for a legitimate authentication server. We introduce a novel Stackelberg game model which captures the essential elements of this interaction between a defender and an offline attacker. In the game the defender first commits to a key-stretching mechanism, and the offline attacker responds in a manner that optimizes his utility (expected reward minus expected guessing costs). We then introduce Cost Asymmetric Secure Hash (CASH), a randomized key-stretching mechanism that minimizes the fraction of passwords that would be cracked by a rational offline attacker without increasing amortized authentication costs for the legitimate authentication server. CASH is motivated by the observation that the legitimate authentication server will typically run the authentication procedure to verify a correct password, while an offline adversary will typically use incorrect password guesses. By using randomization we can ensure that the amortized cost of running CASH to verify a correct password guess is significantly smaller than the cost of rejecting an incorrect password.  Using our Stackelberg game framework we can quantify the quality of the underlying CASH running time distribution in terms of the fraction of passwords that a rational offline adversary would crack. We provide an efficient algorithm to compute high quality CASH distributions for the defender. Finally, we analyze CASH using empirical data from two large scale password frequency datasets. Our analysis shows that CASH can significantly reduce (up to $50\%$) the fraction of password cracked by a rational offline adversary.

\end{abstract}





\section{Introduction} \label{sec:Introduction}
In recent years the authentication servers at major companies like eBay, Zappos, Sony, LinkedIn and Adobe~\FullVersion{\footnote{For example, see \url{ http://www.privacyrights.org/data-breach/} (Retrieved 9/1/2015).}}{\cite{breach:ebay,breach:Zappos,breach:sony,breach:linkedin,breach:IEEE,breach:Adobe}} have been breached. These breaches have resulted in the release of the cryptographic hashes of millions of user passwords,  each of which has significant economic value to adversaries~\cite{passwordBlackMarket,zonenberg2009distributed}. An adversary who has obtained the cryptographic hash of a user's password can mount a fully automated attack to crack the user's password by comparing this hash value to the cryptographic hashes of likely password guesses~\cite{JTR}. This offline attacker can try as many password guesses as he likes; he is only limited by the resources that he is willing to invest to crack the password. 

Offline attacks are becoming increasingly dangerous due to a combination of several different factors. First, improvements in computing hardware make password cracking cheaper (e.g., ~\cite{zonenberg2009distributed}).  Second, empirical data indicates that many users tend to select low entropy passwords~\cite{mostPopularPasswords2012,bonneau2012science,seeley1989password}.  Finally, offline adversaries now have a wealth of training data available from previous password breaches~\cite{PasswordCrackingArticle} so the adversary often has very accurate background knowledge about the structure of popular passwords. 

Password hash functions like PBKDF2\cite{kaliski2000pkcs}, BCRYPT\cite{bcrypt}, Argon2\cite{Argon2} and SCRYPT\cite{percival2012scrypt} employ key-stretching~\cite{morris1979password} to make it more expensive for an offline adversary to crack a hashed password. While key-stretching may reduce the number of password guesses that the adversary is able to try, the legitimate authentication server faces a basic trade-off: he must also pay an increased cost every time a user authenticates.

The basic observation behind our work is that it is possible for the legitimate authentication server to use randomization to gain an advantage in this cat-and-mouse game. The offline adversary will spend most of his time guessing incorrect passwords, while the authentication server will primarily authenticate users with correct passwords. Therefore, it would be desirable to have an authentication procedure whose cost is asymmetric. That is the cost of rejecting an incorrect password is greater than the cost of accepting a correct password. This same basic observation lay behind Manber's proposal to use secret salt values (e.g., ``pepper")~\cite{manber1996simple}. For example, the server might store the cryptographic hash $\Hash(pwd,t)$ for a uniformly random value $t \in \{1,\ldots,m\}$ called the ``pepper". An offline adversary will need to compute the hash function $m$ times in total to reject an incorrect password $pwd'$, while the legitimate authentication server will only need to compute it $\frac{m+1}{2}$ times on average to accept a correct password. 

We introduce Cost Asymmetric Secure Hash (CASH) a mechanism for protecting passwords against offline attacks while minimizing expected costs to the legitimate authentication server. CASH may be viewed as a simple, yet powerful, extension of~\cite{manber1996simple} in which the distribution over $t$ is not-necessarily uniform ---  the ``peppering" idea of Manber~\cite{manber1996simple} is a special case of our mechanism in which the distribution over $t$ is uniform. 

In this paper we seek to address the following questions: How can we quantify the security gains (losses) from the use of secret salt values? What distribution over the secret salt value (t) is optimal for the authentication server? Is there an efficient algorithm to compute this distribution? Does CASH perform better than ``pepper" or deterministic key stretching?

\paragraph{Contributions} We first introduce a Stackelberg (leader-follower) game which captures the essential aspects of our password setting. Our Stackelberg model can provide helpful guidance for the authentication server by predicting whether or not (a particular level of) key-stretching will significantly reduce the number of passwords that would be cracked by a rational offline adversary in the event of a server breach. In our Stackelberg game the authentication server (leader) first commits to a password hashing strategy, and the offline adversary (follower) gets to play his best response to the server's (leader's) action. That is the adversary selects a threshold $B$ and begins guessing passwords until he either 1) cracks the user's password, or 2) gives up after expending $B$ units of work. The adversary will select a threshold $B$ that maximizes his utility.\footnote{Intuitively, the adversary's utility is his expected reward (the value of a cracked password times the probability he cracks it) minus his expected guessing costs (given by the expected number of times that the adversary needs to evaluate the hash function before he succeeds or gives up).}

Next we give an efficient algorithm for computing good strategies for the leader (authentication server) in this Stackelberg game. The defender wants to find a distribution $\tilde{p}_1 \geq \ldots \geq \tilde{p}_m \geq 0$ over the secret running time parameter $t \in \{1,\ldots,m\}$, which minimizes the number of passwords that an offline adversary would crack. When choosing this distribution, the defender is given a constraint (e.g., $\E[t] = \sum_{t=1}^m t\cdot \tilde{p}_t \leq C_{max}$) bounding the server's amortized authentication costs. 

Unfortunately, there are no known polynomial time algorithms to compute the Stackelberg equilibrium of our game as this problem reduces to a non-convex optimization problem.\footnote{By contrast, fixing any CASH distribution $\tilde{p}_i  \geq \ldots \geq  \tilde{p}_m$ it is easy to compute the adversary's best response.} However, we develop an efficient algorithm to solve a closely related goal:  find the  CASH distribution which minimizes the success rate of an adversary with a fixed budget $B$ per user. While this new goal is not equivalent to the Stackelberg equilibrium our experimental results indicate that the resulting CASH distributions translate to good strategies in the original Stackelberg game.  At a technical level we show that this new optimization problem can be expressed as a linear program. The key technical challenge in solving this linear program is that it has exponentially many constraints. Fortunately, this linear program can still be solved in polynomial time using an efficient separation oracle that we develop. We also develop a practical algorithm which can quickly find the (approximately) optimal CASH distribution against a budget $B$ adversary. The algorithm is efficient enough to run on large real world instances (e.g., a dataset of $70$ million passwords). 

Finally, we evaluated CASH using password frequency data from the RockYou password breach and from a (perturbed) dataset of 70 million Yahoo! passwords\cite{bonneau2012science,blocki2016differentially}. Our analysis shows that CASH significantly outperforms the traditional (deterministic) key-stretching defense as well as the ``peppering" defense of \cite{manber1996simple}.  In some instances, CASH reduced the fraction of passwords cracked by a rational adversary by about $50\%$ in comparison to both pepper and traditional key-stretching algorithms.


\section{Background} \label{sec:Background}

Before we introduce the basic CASH mechanism it is necessary to introduce some notation (Section \ref{subsec:Notation}) and review the traditional password based authentication process (Section \ref{subsec:Traditional}).

\subsection{Notation.} \label{subsec:Notation}  We use $\Hash$ to denote a cryptographic hash function and we let $\mathbf{Cost}\left(\mathbf{H} \right)$ denote the cost of evaluating $\mathbf{H}$ one time. To simplify the presentation we will assume that all other costs have been scaled so that $\mathbf{Cost}\left(\mathbf{H} \right)=1$. We use $\Hash^k$ to denote a hash function that is $k$-times as expensive to compute.\footnote{In this work we will not focus on the lower level issue of which key-stretching techniques are used. However, this is an important research area~\cite{PHC} and we would strongly advocate for the use of modern key-stretching techniques like memory hard functions. BCRYPT~\cite{bcrypt} and PBKDF2~\cite{kaliski2000pkcs}, use hash iteration for key-stretching. In this case the cost parameter $k$ specifies the number of hash iterations. For example, if $k=2$ the authentication server would store the tuple $\left(k=2,\Hash\big(\Hash(pwd)\big)\right)$. The disadvantage to this approach is that a hash function $\Hash$ might cost orders of magnitude less to evaluate on an Application Specific Integrated Circuit than it would cost to evaluate on a more traditional architecture.  By contrast, memory costs tend to be relatively stable across different architectures~\cite{DGN03}, which  motivates the use of memory hard functions for password hashing~\cite{Per09}.  Argon2~\cite{Argon2}, winner of the recently completed password hashing competition~\cite{PHC}, and SCRYPT~\cite{percival2012scrypt} use memory hard functions to perform key-stretching. In this paper we will simply use $\Hash^k$ is $k$-times as expensive to compute without worrying about the specific key-stretching techniques that were employed to achieve this property. }  We use $\PasswordSpace$ to denote the space of passwords that users may select, and we use $n$ to denote the number of passwords in this space. We use $p_i$ to denote the probability that a random user selects the password $pwd_i \in \PasswordSpace$. For notational convenience, we assume that the passwords have been sorted so that  $p_1 \geq ... \geq p_n$. Given a set $S$ we will write $x \stackrel{\$}{\gets} S$ to denote a uniformly random sample from the set $S$.  

Table \ref{tab:Notation} contains a summary of the notation used throughout this paper. Some of this notation will be introduced later in the paper when it is first used.

\subsection{Traditional Password Authentication.} \label{subsec:Traditional}
We begin by giving a brief overview of the traditional password authentication process. Suppose that a user registers for an account with username $u$ and password $pwd_u \in \PasswordSpace$. Typically, an authentication server will store a record like the following $\left(u,s_u,k,\mathbf{H}^k\left(pwd_u,s_u \right)\right)$. Here, $s_u \stackrel{\$}{\gets}\{0,1\}^L$ is a random $L$-bit salt~\cite{salt} value used to prevent rainbow table attacks~\cite{rainbowTable}  and the parameter $k$ controls the cost of the hash function. We stress that the salt value $s_u$ and the cost parameter $k$ are stored on the server in the clear so an adversary who breaches the authentication server will learn both of these values. We use the notation $s_u$ to emphasize that this salt value is different for each user $u$. The parameter $k$ is selected subject to the constraint that $k \leq C_{max}$ --- the maximum amortized cost that the authentication server is willing to incur for authentication.\footnote{In the traditional (deterministic) key-stretching setting it is clear the hash cost parameter $k=C_{max}$ is equivalent to the maximum authentication cost parameter $C_{max}$. However, this equivalence will not hold one we introduce a randomized running time parameter $t$. Thus, it is helpful to use separate notation to separate these distinct parameters. }

When the user authenticates he will type in his username $u$ and a password guess $pwd_u'\in \PasswordSpace$. The authentication server first finds the record $\left(u,s_u,k,\mathbf{H}^k\left(pwd_u,s_u \right)\right)$. It then computes $\mathbf{H}^k\left(pwd_u',s_u \right)$ and verifies that it matches the stored hash value $\mathbf{H}^k\left(pwd_u,s_u \right)$. Note that authentication will always be successful when the user's password is correct (e.g., $pwd_u' = pwd_u$) because the hash values $\mathbf{H}^k\left(pwd_u',s_u \right)$ and $\mathbf{H}^k\left(pwd_u,s_u \right)$ must match in this case. Similarly, if the user's password is incorrect (e.g., $pwd_u \neq pwd_u'$) then authentication will fail with high probability because the cryptographic hash function $\mathbf{H}$ is collision resistant. \newline

\noindent {\bf Server Cost.} Under this traditional password mechanism the cost of verifying/rejecting a password is simply $k$. The authentication server can increase guessing costs for an offline adversary by increasing $k$, but in doing so the authentication server will increase its own authentication costs proportionally. 

\noindent{\bf Authentication Time Increase.} By increasing the cost parameter $k$ the authentication server might potentially increase delay times for the user --- especially if key-stretching is performed on a sequential computer. Bonneau and Schechter\cite{BS14} estimated that $\mathbf{Cost}\left(\mathbf{H}\right) \approx \$ 7 \times 10^{-15}$  for the SHA-256 hash function based on observations of the Bitcoin network. A modern CPU can evaluate SHA-256 around $10^7$ times per second so an authentication server who uses hash iteration for key-stretching would need to select $k \leq 10^7$ if he wants to ensure that user delay is at most one second. In this case we would seem to have an upperbound $\mathbf{Cost}\left(\mathbf{H}^k\right) \leq \$7\times 10^{-8}$ on the cost of a hash function that can be evaluated in $1$ second. Fortunately, this bound only applies to naive hash iteration\footnote{As we previously noted hash iteration alone is not a particularly effective key-stretching technique. The cost of computing SHA-256 can be reduced by a factor of about 1 million on customize hardware --- e.g., see \url{https://bitcoinmagazine.liberty.me/bitmain-announces-launch-of-next-generation-antminer-s7-bitcoin-miner/} (Retrieved 5/4/2016). Furthermore, we note that modern Bitcoin miners already use Application Specific Integrated Circuits to compute SHA-256 so the upper bound from \cite{BS14} implicitly incorporates this dramatic cost reduction. By contrast, the adversary cannot (significantly) reduce the cost of evaluating a memory hard function by developing customized hardware.}. More effective key-stretching techniques could be used to increase $\CostHk$ by several orders of magnitude (e.g., $\CostHk \geq \$10^{-5}$) without imposing longer authentication delays on the user (even if key-stretching is performed on a sequential computer). For example, the SCRYPT~\cite{percival2012scrypt} and Argon2~\cite{Argon2} hash functions were intentionally designed to use a larger amount of memory so that it is not possible to (significantly) reduce hashing costs by developing customized hardware. Additionally, Argon2~\cite{Argon2}, winner of the password hashing competition, has an optional parameter that would allow the authentication server to exploit parallelism to further reduce the amount of time necessary to perform key-stretching. 


\subsection{Adversary Model} We consider an untargeted offline attacker whose goal is to break as many passwords as possible. An offline attacker has breached the authentication server and has access to all of the data stored on the server. In the traditional authentication setting an offline adversary learns the tuple $\left(u,s_u,k,\Hash^k\left(pwd_u,s_u \right)\right)$ for each user $u$. The adversary will also learn the hash function $\Hash$ since the code to compute $\Hash$ is present on the authentication server. We assume that the adversary only uses $\Hash$ in a blackbox manner (e.g., the adversary can query $\Hash$ as a random oracle, but he cannot invert $\Hash$). In general we assume the adversary will obtain the source code for any other procedures that are used during the authentication process. While the authentication server can limit the number of guesses that an online adversary can make (e.g., by locking the adversary out after three incorrect guesses), the authentication server cannot directly limit the number of guesses that an offline attacker can try. An offline attacker is limited only by the resources that s/he is willing to invest trying to crack the user's password.

We assume that the adversary has a value $v_u$ for cracking user $u$'s  password.  An untargeted offline attacker has the same value $v_u = v$ for every user $u$. Symantec recently reported that passwords sell for between $\$4$ and $\$30$ on the black market~\cite{passwordBlackMarket} so we might reasonably estimate that $v \in \left[\$4,\$30\right]$.\footnote{However, this estimate of the adversary's value could be too high because it does not account for the inherent risk of getting caught when selling/using the password}

We also assume that the adversary knows the empirical password distribution $p_1 \geq ... \geq p_n$ over user selected passwords as well as the corresponding passwords $pwd_1,\ldots, pwd_n$. Thus, the adversary knows that a random user will select $pwd_1$ with probability $p_1$, but the adversary does not know which users selected $pwd_1$. 

The adversary will select a threshold $B$ and check (up to) $B$ passwords. In this case the fraction of passwords that the offline adversary will break is at most $\sum_{i=1}^{B} p_i$. Equality holds when the offline adversary adopts his optimal guessing strategy and checks the $B$ most likely passwords $pwd_1,\ldots,pwd_{B}$. In this case the adversary's utility would be
\[ \UAdvDet{B}{v}{k} \doteq  v \sum_{i=1}^{B} p_i - \left(k\sum_{i=1}^{B} i \cdot p_i + \sum_{i=B+1}^n B\cdot p_i\right) \ . \]
The first term is the adversary's expected reward. The last term is the adversary's expected guessing cost.\footnote{Note that for $i \leq B$ the adversary finishes early after only $i$ guesses if and only if the user selected password $pwd_i$ (probability $p_i$). If the user selected password $pwd_i$ with $i > B$ then the adversary will quit after $B$ guesses.} Let $B^* = B_{v}^{det,*} = \arg\max_{B} \UAdvDet{B}{v}{k}$ denote the adversary's utility optimizing strategy. Then the fraction of passwords cracked by a rational adversary will be
\begin{equation}\label{eq:PAdvDet} \PAdvDet{v}{k} \doteq  \sum_{i=1}^{B^*} p_i \ . \end{equation}

\cut{
We use $\mathcal{P}_{Adv,B}$ to denote the fraction of user accounts that a budget-$B$ adversary can compromise in an offline attack.  The value $\mathcal{P}_{Adv,B}$ will depend on the mechanisms used by the authentication server. 

\noindent {\bf Example.} If the authentication server used the traditional deterministic hash iteration mechanism (Section \ref{subsec:Traditional}) with a parameter $k$ then a budget-$B$ adversary will be able to try $\lfloor B/k\rfloor$ password guesses. The adversary's optimal strategy would be to try the $\left\lfloor B/k\right\rfloor$ most popular passwords $pwd_1,\ldots, pwd_{\lfloor B/k \rfloor}$. His success rate would be \[\mathcal{P}_{Adv,B} = \sum_{i=1}^{\left\lfloor B/k \right\rfloor} p_i \ . \]
Thus, the authentication server can reduce the value of $\mathcal{P}_{Adv,B}$ by increasing $k$, but this also increases costs for the authentication server. 

Henceforth, we assume for simplicity that the cost units have been scaled so that $\mathbf{Cost}\left(\mathbf{H} \right)=1$. This allows us to simplify the presentation by eliminating  $\mathbf{Cost}\left(\mathbf{H} \right)$ from our expressions.}

\begin{table}[!t]
\renewcommand{\arraystretch}{1.3}
\caption{Notation}
\centering
\begin{tabular}{|p{1.1in}|p{4.3in}|}
\hline 
{\bf Term} & {\bf Explanation} \\
\hline
$\PasswordSpace$ & space of passwords \\
\hline 
$n$ & number of passwords in $\PasswordSpace$ \\
\hline
$pwd_i$ & the $i$'th most likely password in $\PasswordSpace$ \\
\hline 
$p_i$ &  probability that a random user selects $pwd_i$ \\
\hline
$m$ & the number of evaluations of $\mathbf{H}^k$ necessary to reject an incorrect password using CASH \\
\hline
$t \in \{1,\ldots,m\}$ & hidden running time parameter which specifies the running time of CASH when verifying an correct password. $t$ is randomly selected during account creation. \\
\hline
$\tilde{p}$ & a distribution over the hidden running time parameter $t$ \\
\hline
$\tilde{p}_j$ & the probability that the running time parameter is $t = j$ \\
\hline
$\pi_{i}$ & the probability of the $i$'th most likely  tuple $(pwd,t)$\\
\hline
%
$\alpha$ & probability of seeing a correct password in a random authentication session \\
\hline
$\Hash$ & a cryptographic hash function with $\CostH=1$ \\
\hline
 $\Hash^k$ & a cryptographic hash function with $\CostHk=k$  \\
\hline
$\ServerCost{\alpha}$ & $mk(1-\alpha)+\alpha k \sum_{t=1}^m t\cdot\tilde{p}_t$, the amortized cost of a random authentication session. \\ 
\hline
$C_{max}$ & the maximum (amortized) cost that the authentication server is willing to incur per authentication  \\
\hline
$v$ & adversary's true value for a cracked password \\
\hline 
$\hat{v}$ & the authentication server's estimate for $v$ \\
\hline
$$\PAdvCASH{v}{\hat{v}}{C}$$ & the fraction of passwords cracked by a rational value $v$ adversary, when the authentication server optimizes the CASH distribution $\tilde{p}$ under the belief $\hat{v}$ subject to the cost constraint $\ServerCost{\alpha} \leq C_{max}$. \\
\hline
$$\PAdvUnif{v}{C}$$ & the fraction of passwords cracked by a rational value $v$ adversary, when the authentication server uses the uniform distribution $\tilde{p}_i = 1/m$. The hash cost parameter $k$ is now tuned subject to the cost constraint $\ServerCost{\alpha} \leq C$. \\
\hline
$$\PAdvDet{v}{C}$$ & the fraction of passwords cracked by a rational value $v$ adversary when the authentication server uses deterministic key-stretching techniques. The hash cost parameter is set to $k=C$ so that the servers cost is $C$ for each authentication session. \\
\hline
\end{tabular}

\label{tab:Notation}
\end{table}

\section{CASH Mechanism} \label{sec:CASH}
In this section we introduce the basic CASH mechanism, while deferring until later the question of how to optimize the parameters of the mechanism. 

\subsection{CASH Authentication.} 
Observe that in traditional password authentication the costs of verifying and rejecting a password guess are symmetric. The goal of CASH is to redesign the authentication mechanism so that these costs are not symmetric. In particular, we want to ensure that the cost of rejecting an incorrect password is greater than the cost of accepting a correct password. This is a desirable property because most of the adversary's password guesses during an offline attack will be incorrect. By contrast, the authentication server will spend most of its effort authenticating legitimate users. 
\subsubsection{Creating an Account} Suppose that a user $u$ registers for an account with the password $pwd_u \in \PasswordSpace$. In CASH authentication the authentication server stores the value $\left(u,s_u,k,\mathbf{H}^{k}\left(pwd_u,s_u,t_u \right)\right)$. As before $s_u$ is a random salt value and $k$ is the number of hash iterations. The key difference is that we select a random value $t_u$ from the range $\{1,\ldots,m\}$. We stress that the value $t_u$ is not stored on the authentication server (unlike the salt value $s_u$). Thus, the value $t_u$ will not be available to an adversary who breaches the server. \FullVersion{}{The account creation process is formally presented in Algorithm \ref{alg:CASHCreateAccount}.} We use the notation $t_u$ here to emphasize that this value is chosen independently for each user $u$. Intuitively, the parameter $t_u$ specifies the number of times that the authentication server needs to compute $\mathbf{H}^k$ when verifying a correct password guess using CASH.

\subsubsection{Authentication} When the user $u$ tries to authenticate using the password guess $pwd_u'$ the authentication server first locates the record $\left(u,s_u,k,\mathbf{H}^{k}\left(pwd_u,s_u,t_u \right)\right)$. The authentication server then computes $\mathbf{H}^{k}\left(pwd_u',s_u,t \right)$ for each value $t \in \{1,\ldots, m\}$. Authentication is successful if the hashes match for {\em any} value $t \in \{1,\ldots,m\}$. This is guaranteed to happen after $t_u$ steps whenever the user's password is correct ($pwd_u'=pwd_u$), and this is highly unlikely whenever the user's password is incorrect. \FullVersion{A more formal presentation of the authentication process can be found in the full version\cite{fullVersion} of this paper.}{The authentication process is formally presented in Algorithm \ref{alg:CASHAuthenticate}.}

\subsubsection{CASH Notation} We use $\tilde{p}_i$ to denote the probability that we set $t_u = i$ during the account creation process. For notational convenience we will assume that these values are sorted so that $\tilde{p}_1 \geq \ldots \geq \tilde{p}_m$. We will use $t \leftarrow \tilde{p}$ to denote a random sample from $\{1,\ldots,m\}$ in which $\Pr_{t \leftarrow \tilde{p}}\left[t=i\right] = \tilde{p}_i$.  For now we assume that the CASH distribution $\tilde{p}$ is given to us. In later sections we will discuss how to select a good distribution $\tilde{p}$.


\FullVersion{}{
\begin{algorithm}[H] 
 \caption{CASH: Create Account}
\label{alg:CASHCreateAccount}
\begin{algorithmic}[1]
 \renewcommand{\algorithmicrequire}{\textbf{Input:}}
 \renewcommand{\algorithmicensure}{\textbf{Output:}}
\REQUIRE $u$, $pwd_u$, $\tilde{p}=\left(\tilde{p}_1,\ldots, \tilde{p}_{m}\right)$, $k$, $L$
\STATE $s_u \stackrel{\$}{\gets}\{0,1\}^L$
\STATE $t_u \gets \tilde{p}$ 
\STATE $h \gets \mathbf{H}^k\left(pwd_u,s_u,t_u \right)$
\STATE $\mathbf{StoreRecord}\left(u,s_u,k,h \right)$ 
\end{algorithmic}

\end{algorithm}

\begin{algorithm}[H] 
 \caption{CASH:Authenticate}
\label{alg:CASHAuthenticate}
\begin{algorithmic}[1]
 \renewcommand{\algorithmicrequire}{\textbf{Input:}}
 \renewcommand{\algorithmicensure}{\textbf{Output:}}
\REQUIRE  $u$, $pwd_u$
\STATE $R \gets \mathbf{TryFindRecord}\left(u\right)$
\IF{$R = \emptyset$}
\RETURN ``Username Not Found."
\ENDIF
\STATE $\left(u,s_u,k,h \right) \gets R$
\FOR{$t = 1, \ldots,m$ }
\STATE $h_t \gets \mathbf{H}^k\left(pwd_u,s_u,t \right)$
\IF{$h_t = h$}
\RETURN ``Authentication Successful"
\ENDIF
\ENDFOR
\RETURN ``Authentication Failed"
\end{algorithmic}

\end{algorithm}
}

\subsection{Cost to Server}
 The cost of rejecting an incorrect password guess is $m\cdot k$ because the server must evaluate $\mathbf{H}^{k}\left(pwd_u,s_u,t_u \right)$ for all $m$ possible values of $t_u \in \{1,\ldots, m\}$. However, whenever a password guess is correct the authentication server can halt computation as soon as it finds a match, which  will happen after $t_u$ iterations. Here, we assume that the authentication server will minimize its amortized cost by trying the most likely values of $t_u$ first. If we let $\alpha$ denote the probability that the user enters his password correctly during a random authentication session then the amortized cost of the authentication server is \[\ServerCost{\alpha} \doteq \left(1-\alpha\right) k\cdot m  + \alpha\cdot k \sum_{i=1}^m i \cdot\tilde{p}_i\ .\]
In general, we will assume that the server has a maximum amortized cost $C_{max}$ that it is willing to incur for authentication.\footnote{For example, $C_{max}$ might be (approximately) given by the maximum computational load that the authentication server(s) can handle divided by the maximum (anticipated) number of users authenticating at any given point in time.} Thus, the authentication server must pick the distribution $\tilde{p}$ subject to the cost constraint $\ServerCost{\alpha} \leq C_{max}$.

\subsection{Adversary Response}
Fixing the CASH distribution $\tilde{p}$ induces a distribution over pairs $(pwd,t)\in \PasswordSpace \times \{1,\ldots,m\}$, namely $\Pr[(pwd,t)] = p_i\cdot \tilde{p}_t$. Once the adversary selects a threshold $B$ the adversary's optimal strategy is to try the $B$ most likely pairs. In this case the adversary's utility will be 
\begin{equation} \label{eq:UAdvCASH}  \UAdvCASH{B}{v} = v \sum_{i=1}^B \pi_i - k\sum_{i=1}^B i\cdot \pi_i - k\sum_{i=B+1}^{mn} B\cdot \pi_i \ , \end{equation}
where the terms $\pi_{1} \geq \ldots \geq \pi_{mn}$ denote the probabilities of each pair $(pwd,t)\in \PasswordSpace \times \{1,\ldots,m\}$  (in sorted order). 
In general, the distribution $\tilde{p}$ that the authentication server selects may depend on the maximum (amortized) server cost $C_{max}$ as well as our belief $\hat{v}$ about the adversary's value for a cracked password. Once $\hat{v}$ and $C_{max}$ (and thus $\tilde{p}_1,\ldots,\tilde{p}_m$, $k$ and $\pi_1,\ldots,\pi_{mn}$) have been fixed we can let $B^* = B_{v}^{CASH,*} = \arg\max_{B} \UAdvCASH{B}{v}$ denote the adversary's utility optimizing response. Then the fraction of passwords cracked by a rational adversary will be
\begin{equation} \label{eq:PAdvCASH} \PAdvCASH{v}{\hat{v}}{C_{max}} \doteq  \sum_{i=1}^{B^*} \pi_i \ . \end{equation}
Similarly, we will use $\PAdvUnif{v}{C_{max}}$ to denote the fraction of passwords cracked by a rational adversary when $\tilde{p}$ is the uniform distribution.\footnote{Note that $\PAdvUnif{v}{C_{max}}$ does not depend on $\hat{v}$, our belief about the adversary's value, because the choice of $\tilde{p}$ (and $k$) is independent of this belief.} In this case the hash cost parameter $k$ is tuned to to ensure that $\ServerCost{\alpha} \leq C_{max}$ --- this can be achieved when \begin{equation} \label{eq:kUnifCash} k = \frac{C_{max}}{(1-\alpha)m + \alpha \left(\frac{m+1}{2} \right)} \ . \end{equation}

\subsubsection{Example Distribution} One simple, yet elegant, way to achieve the goal of cost asymmetry is to set $\tilde{p}_j = \frac{1}{m}$ for each $j  \in \{1,\ldots, m\}$~\cite{manber1996simple}. We will sometimes call this solution uniform-CASH in this paper because it is a special case of the CASH mechanism. The amortized cost of verifying a correct password guess with uniform-CASH is $C_{SRV,1} =k \left(\frac{m+1}{2} \right) $. By contrast, the cost of rejecting an incorrect password guess is $k\cdot m $ --- approximately twice the cost of verifying a correct password guess.

\paragraph{Examples with Analysis} The above mechanism can already be used to significantly reduce the fraction of user passwords that would be cracked in an offline attack. We demonstrate the potential power of CASH with two (simplistic) examples. To keep the examples simple we will assume that that users never forget or mistype their passwords (i.e., $\alpha=1$). In the first example, every user selects one of two passwords (e.g., $pwd_1=$``123456" and $pwd_2$=``iloveyou") with probability $p_1=2/3$ and $p_2=1/3$ respectively, and the untargeted adversary has a value of $v=\frac{4}{3} C_{max} + \epsilon$, just slightly more than $C_{max}$ --- the amortized cost  incurred by the authentication server during an authentication session. 
\begin{itemize}
\item (Deterministic Key-Stretching) The defender sets the hash cost parameter $k=C_{max}$ and stores the deterministic hash value $\Hash^k$. It is easy to check that the adversary's optimal response is to choose the maximum threshold $B^*=2$. In this case the adversary cracks the password with probability $\PAdvDet{v}{C_{max}} = 1$. 
\item (Uniform CASH) The defender sets $\tilde{p}_i=\frac{1}{m}$ for each $i$ and he selects cost parameter $k = 2\cdot C_{max}/(m+1)$ to ensure that $C_{SRV,1} \leq C_{max}$ --- see eq \ref{eq:kUnifCash} . It is not too difficult to see that the adversary's optimal response is to choose the threshold $B^* = 0$ (i.e., give up without guessing).\footnote{In particular, if the adversary instead sets $B^*=2m$ (i.e., keep guessing until he succeeds) then his expected guessing costs will be $p_1k\left(\frac{m+1}{2} \right)+(1-p_1)k\left(m+\frac{m+1}{2}  \right) = (1-p_1) km + \frac{m+1}{2}k = C_{max}  + \frac{1}{3}\left(\frac{2mC_{max}}{m+1} \right) = \frac{5C_{max}}{3}-\frac{2C_max}{m+1}> v$.} Thus, $\PAdvUnif{v}{C_{max}} = 0  < 1 = \PAdvDet{v}{C_{max}}$.

\end{itemize}

This first example illustrates the potential advantage of randomization. The next example illustrates the potential advantage of non-uniform distributions. Example 2 is the same as example 1 except that we increase the adversary's value to $v=\frac{5}{3}C_{max}$.  
\begin{itemize}
\item (Deterministic Key-Stretching) Increasing $v$ can only increase $\PAdvDet{v}{C_{max}}$. Thus,  $\PAdvDet{v}{C_{max}}=1$.
\item (Uniform CASH) Now the adversary's optimal strategy is to choose the maximum threshold $B^* = 2m$ (i.e., keep guessing until he finds the password). Thus, $\PAdvUnif{v}{C_{max}} = 1$.
\item (non-uniform CASH) Suppose that the authentication server, knowing that $\hat{v}=v = \frac{5\cdot C_{max}}{2}$, sets $m=5$, $k = C_{max}/2$ and sets $\tilde{p}_1 = 9/16,\tilde{p}_2=\tilde{p}_3=\tilde{p}_4=1/8$ and $\tilde{p}_5=1/16$.\footnote{It is easy to verify that $\ServerCost{\alpha}=2k = C_{max}$. } In this case it is possible to verify that the adversary's optimal response is to set $B^*=2$ meaning that the adversary will try guessing the two most likely pairs $(pwd_1,t=1)$ and $(pwd_2,t=1)$  before giving up. Thus, $\PAdvCASH{v}{\hat{v}}{C_{max}} = \big(p_1+p_2\big)\tilde{p_1} = \frac{9}{16} < 1$. 
\end{itemize}

Admittedly these example are both  overly simplistic. However, we will later consider several empirical password distributions and demonstrate that non-uniform CASH distributions are often significantly better than both uniform CASH and deterministic key-stretching.

\section{Stackelberg Model} \label{sec:Definitions}
In the last section we observed that uniform CASH can reduce the adversary's success rate compared to deterministic key-stretching techniques with comparable costs. We also saw that sometimes it is possible to do even better than uniform CASH by selecting a non-uniform distribution over $t$.\footnote{Of course in some cases the uniform distribution might still be optimal.} This observation leads us to ask the following question: What distribution over $t$ leads to the optimal security results? 

In this section we first formalize the problem of finding the optimal CASH distribution parameters $\tilde{p}_1 \geq \ldots \geq \tilde{p}_m \geq 0$. Intuitively, we can view this problem as the problem of computing the Stackelberg equilibria of a certain game between the authentication server and an untargeted offline adversary. Stackelberg games and their applications have been an active area of research in the last decade (e.g., ~\cite{conitzer2006computing,jain2010security,blocki2013audit,yin2010stackelberg}). 
For now we will simply focus on formulating this goal as an optimization problem. In later sections we will present a polynomial time algorithm to good solutions to this optimization problem (Sections \ref{sec:Algorithm} and \ref{sec:PracticalAlgorithm}) and we will evaluate this algorithm on empirical password datasets (Section \ref{sec:Evaluation}).


\cut{In Section \ref{sec:Algorithm} we present a polynomial time algorithm to find this optimal distribution given the values $p_1,\ldots, p_n$ and $B$. Finally, in Section \ref{sec:Evaluation} we use our algorithm to answer our question. No, the uniform distribution is not always optimal --- especially when the distribution over passwords is far from uniform. By optimizing the distribution $\tilde{p}_1,\ldots, \tilde{p}_m$ we can achieve an even greater reduction in the fraction of passwords that an adversary could crack in an offline attack --- without increasing costs for a legitimate server.}

Before the Stackelberg game begins the adversary is given a value $v$ for cracked passwords and the authentication server is given an honest estimate $\hat{v}=v$ of the adversary's value.\footnote{In the game the authentication server will assume that $\hat{v}$ is indeed the correct value when he computes the distribution $\tilde{p}$. Of course, in our empirical analysis we will also be interested in exploring how CASH performs when this estimate is incorrect $\hat{v}\neq v$.} The authentication server is also given a bound $C_{max}$ on the expected cost of an authentication round.

\paragraph{Defender Action} The authentication server (leader) moves first in our Stackelberg game. The authentication server must commit to a CASH distribution  $\tilde{p}$ and a hash cost parameter $k$. The values must be selected subject to a constraint on the maximum amortized cost for the authentication server
$$\ServerCost{\alpha}=\left(1-\alpha\right)m \cdot k + \alpha\cdot k\sum_{i=1}^m \left( i\cdot \tilde{p}_i\right) \leq C_{max} \ . $$
Intuitively, we can view the value $\alpha$ as being given by nature and the parameter $C_{max}$ is given by the computational resources of the authentication server. 

\paragraph{Offline Adversary} After the authentication server commits to $\tilde{p}$ and $k$ the offline adversary is given access to all of the hashed passwords stored on the authentication server. The adversary can try guesses of the form $\left( pwd_i,j\right)$. This particular guess is correct if and only if the user $u$ selected password $pwd_u = pwd_i$ and we selected the secret salt value $t_u = j$. For an untargeted attacker the probability that this guess is correct is $p_i\cdot \tilde{p}_j$. We can describe the action of a rational adversary using a threshold $B$ which denotes the maximum number of pairs $(pwd,t)$ that he will check (equivalently the maximum number of times he will compute $\Hash^k$). Intuitively, we don't need to specify which pairs the adversary guesses because a rational adversary will always check the $B$ most likely pairs.

We remark that we assume that an offline attacker will be able to obtain the CASH parameters $\tilde{p}_1,\ldots,\tilde{p}_m$ and $k$ that we select.\footnote{An offline adversary has already breached authentication server which will contain code to sample $t_u$  whenever a new user $u$ creates an account.} The adversary also knows the empirical password distribution $p_1 \geq \ldots \geq p_n$ and the associated passwords $pwd_1,\ldots,pwd_n$. 



\paragraph{Optimization Goal} Informally, the defender's goal is to minimize the probability that the rational adversary succeeds in cracking each user's password. The distribution that achieves this goal is the Stackelberg equilibrium of our game. Formally, our optimization goal is presented as Optimization Goal \ref{goal:MinAdvSuccess}. We are given as input the empirical password distribution $p_1,\ldots,p_n$ as well as the value $\hat{v}$ for the adversary,  a maximum cost $C_{max}$ for the authentication server, the CASH parameter $m$ and the fraction $\alpha$ of authentication sessions in which enter their correct password. We want to find values $\tilde{p}_1,\ldots,\tilde{p}_m$ and $k$ that minimize the fraction of cracked passwords $\PAdvCASH{v}{\hat{v}}{C_{max}}$ subject to several constraints. Constraints 1 and 2 ensure that $\tilde{p}_1,\ldots,\tilde{p}_m$ form a valid probability distribution, and constraint 3 ensures that the amortized cost of authentication is at most $C_{max}$. Constraint 4 simply defines the variables $\pi_1,\ldots,\pi_{mn}$ where $\pi_i$ is the probability of the $i$'th most likely tuple $(pwd,t)$. Constraints 5 implies that $B^*$ is the adversary's optimal response (e.g.,  $\UAdvCASH{B^*}{v} \geq \UAdvCASH{B}{v}$ for any other threshold $B$ that the adversary might choose). Finally, $\sum_{i=1}^{B^*} \pi_i$, our minimization goal,   is the fraction of passwords cracked under the adversary's utility optimizing response $B^*$. 

\begin{goal}[H]
\centering
\caption{Minimize Adversary Success Rate}
\label{goal:MinAdvSuccess} 
\begin{algorithmic}
 \renewcommand{\algorithmicrequire}{\textbf{Input Parameters:}}
 \renewcommand{\algorithmicensure}{\textbf{Variables:}}
\REQUIRE$p_1,\ldots, p_n$, $\hat{v}$, $C_{max}$, $m$ and $\alpha$.
\ENSURE   $\tilde{p}_1,\ldots, \tilde{p}_m, \pi_1,\ldots,\pi_{nm}$, $k$
\STATE {\bf minimize} $\sum_{i=1}^{B^*} \pi_i$
subject to
\STATE (1) $1\geq \tilde{p}_1 \geq \ldots \geq \tilde{p}_m \geq 0$,
\STATE (2) $\sum_{i=1}^m \tilde{p}_i = 1$,
\STATE (3) $\left(1-\alpha\right)m k + \alpha k\sum_{i=1}^m \left( i\cdot \tilde{p}_i\right) \leq C_{max}$,
\STATE (4) $\pi_1, \ldots , \pi_{mn}  = \mathbf{Sort}\left( p_{1}\cdot \tilde{p}_1, \ldots, p_n\cdot\tilde{p}_{m}  \right)$, and

\STATE (5) $\forall B \in \{0,1,\ldots,mn\}$ we have 
\[ \UAdvCASH{B^*}{v} \geq \UAdvCASH{B}{v} \ . \]
\end{algorithmic}

\end{goal}

Unfortunately, Optimization Goal \ref{goal:MinAdvSuccess} is inherently non-convex due to the combination of constraints 4 and 5.\footnote{Substituting in the formula for $\UAdvCASH{B}{v}$ constraint 5 becomes $v\sum_{i=1}^B \pi_i - k\sum_{i=1}^B i\cdot \pi_i - k\sum_{i=B+1}^{mn} B\cdot \pi_i$  $\leq v \sum_{i=1}^{B} \pi_i - k\sum_{i=1}^{B^*} i\cdot \pi_i - k\sum_{i=B^*+1}^{mn} B^* \cdot \pi_i$, where $\pi_i$ depends on the $\mathbf{Sort}$ operation.  } Thus, it is not clear whether or not there is a polynomial time algorithm to compute the Stackelberg equilibria. However, as we will see in the next section, there is a polynomial time algorithm to solve a very closely related goal. Minimize the number of passwords that a threshold $B$ adversary can crack  (Goal \ref{goal:MinAdvSuccessAsLP}).

\section{Algorithms} \label{sec:Algorithm}
In this section we show how the goal of minimizing the success rate of a threshold $B$ adversary can be formulated as a linear program with exponentially many constraints (Optimization Goal \ref{goal:MinAdvSuccessAsLP}). We also show that this linear program can be solved in polynomial time by developing an efficient separation oracle. Unfortunately, this polynomial time algorithm is not efficient enough to solve the large real-world instances we consider in our experiments in Section \ref{sec:Evaluation}. However, building on ideas from Section \ref{sec:Algorithm}, we develop a more efficient (in practice) algorithm in Section \ref{sec:PracticalAlgorithm}. This new algorithm always finds an approximately optimal solution to Optimization Goal \ref{goal:MinAdvSuccessAsLP}. While we do not have any theoretical guarantees about its running time, we found that it converged quickly on every instance we tried. Furthermore, as we will see in our experimental evaluation, the algorithm results in significantly improved Stackelberg strategies. 


We remark that our experimental results in Section \ref{sec:Evaluation} can be understood without reading this section. In particular, it is possible to view the algorithms in  Sections \ref{sec:Algorithm} and \ref{sec:PracticalAlgorithm}, as a blackbox heuristic algorithm that finds reasonably good solutions to Optimization Goal \ref{goal:MinAdvSuccess}. 
A more empirically inclined reader may wish to skip to our experimental results in Section \ref{sec:Evaluation} after skimming through this section.

\subsection{LP Formulation} \label{subsec:LP} We first show how to state our goal, minimize the number of passwords that a threshold $B$ adversary will crack, as a linear program. Our LP uses the following variables $\mathcal{P}_{Adv,B}, \tilde{p}_1,\ldots,\tilde{p}_m$. Intuitively, the variable $\mathcal{P}_{Adv,B}$ represents the fraction of passwords that a threshold $B$ adversary can crack. At a high level our Linear Program can be understood as follows: minimize $\mathcal{P}_{Adv,B}$ subject to the requirement that no feasible strategy for the threshold $B$ adversary  achieves a success rate greater than $\mathcal{P}_{Adv,B}$. This requirement can be expressed as a combination of exponentially many linear constraints. Formally, our LP is presented as Optimization Goal \ref{goal:MinAdvSuccessAsLP}. 

\begin{goal}[H]
\caption{Minimize Threshold $B$ Adversary Success Rate}
\label{goal:MinAdvSuccessAsLP} 
\begin{algorithmic}
 \renewcommand{\algorithmicrequire}{\textbf{Input Parameters:}}
 \renewcommand{\algorithmicensure}{\textbf{Variables:}}
\REQUIRE $p_1,\ldots, p_n$, $B$, $C_{max}$, $m$, $k$, $\alpha$
\ENSURE $\tilde{p}_1,\ldots, \tilde{p}_m, \mathcal{P}_{Adv,B} $
\STATE {\bf minimize} $\mathcal{P}_{Adv,B}$ subject to
\STATE (1) $1\geq \tilde{p}_1 \geq \ldots \geq \tilde{p}_m \geq 0$,
\STATE (2) $\sum_{i=1}^m \tilde{p}_i = 1$,
\STATE (3) $\left(1-\alpha\right)m k + \alpha k\sum_{i=1}^m \left( i\cdot \tilde{p}_i\right) \leq C_{max}$,
\STATE (4) $0 \leq \mathcal{P}_{Adv,B} \leq 1$, and
\STATE (5) $\forall S \subset \PasswordSpace \times \{1,\ldots,m\}$ s.t $|S| = B$ we have
$$\mathcal{P}_{Adv,B} \geq \sum_{(i,j) \in S} p_i\cdot \tilde{p}_j \ .$$  
\end{algorithmic}
\end{goal}

The key intuition is that all of the (5) constraints ensure that $\mathcal{P}_{Adv,B}$ is at least at big as the best success rate for a threshold $B$ adversary. This is true because the optimal guessing strategy for a threshold $B$ adversary is to guess the $B$ most likely tuples $(pwd,t)$. Let $S'$ denote these $B$ most-likely tuples then one of the type (5) constraints says that $\mathcal{P}_{Adv,B} \geq \sum_{(i,j) \in S'} p_i\cdot \tilde{p}_j$. Thus, type (5) constraints guarantee we cannot `cheat' by pretending like the adversary will follow a suboptimal strategy (e.g., spending his guessing budget on the least likely passwords) when we solve Optimization Goal \ref{goal:MinAdvSuccessAsLP}.

The key challenge in solving Optimization Goal \ref{goal:MinAdvSuccessAsLP} is that there are exponentially many type (5) constraints. Our main result in this section states that we can still solve this problem in polynomial time.  
\newcommand{\thmPolyTime}{We can find the solutions to Optimization Goal \ref{goal:MinAdvSuccessAsLP} in polynomial time in $m$, $n$ and $L$, where $L$ is the bit precision of our inputs.}
\begin{theorem}\label{thm:Polytime} \thmPolyTime
\end{theorem}

\FullVersion{The proof of Theorem \ref{thm:Polytime} can be found in the full version of this paper\cite{fullVersion}.}{The proof of Theorem \ref{thm:Polytime} can be found in the appendix.} We briefly overview the proof strategy here. The key idea is to build a polynomial time separation oracle for Optimization Goal \ref{goal:MinAdvSuccessAsLP}. Given a candidate solution $\tilde{p}$ the separation oracle should either tell us that the solution is feasible (satisfies all type (5) constraints) or it should  find an unsatisfied constraint.  We can then use the ellipsoid method~\cite{khachiyan1980polynomial} with our separation oracle to solve to solve the linear program in polynomial time.  \FullVersion{}{In appendix \ref{subsec:SeparationOracle} we show how to develop a polynomial time separation oracle for our linear programs.} Intuitively, the separation oracle simply sorts the tuples $\PasswordSpace \times \{1,\ldots, m\}$ using the associated probabilities $\Pr[(pwd_i,t)] = p_i\cdot \tilde{p}_t$. Then we can find the set $S'$ of the $B$ most likely tuples and check to see if the constraint $\mathcal{P}_{Adv,B} \geq \sum_{(i,j) \in S'} p_i\cdot \tilde{p}_j$ is satisfied. 

\cut{As before Constraints (1) and (2) ensure that we get a valid probability distribution over our runtime parameter $t$ and constraint (3) ensures that the amortized cost of verifying a correct password guess is at most $C_{SRV}$. The key difference between Goal \ref{goal:MinAdvSuccessAsLP} and Goal \ref{goal:MinAdvSuccess} is that we treat $\mathcal{P}_{Adv,B}$ as a variable, which represents the adversary's success probability. Constraint (4) ensures that $\mathcal{P}_{Adv,B}$ is a valid probability value. We remark that our linear program would be meaningless without our type (5) constraints because  constraints (1)--(3) do not constrain the value $\mathcal{P}_{Adv,B}$. Thus, without type (5) constraints the optimal solution to our LP would simply set  $\mathcal{P}_{Adv,B} = 0$. We use constraint (5) to ensure that for {\em every} feasible strategy that the adversary might adopt to attack the user's password the value of $\mathcal{P}_{Adv,B}$ is at least as large as the adversary's success rate when following that strategy. In other words type (5) constraints guarantee we cannot cheat by pretending like the adversary will follow a suboptimal strategy (e.g., spending his guessing budget on the least likely passwords).}

Once we have a polynomial time algorithm to solve Optimization Goal \ref{goal:MinAdvSuccessAsLP} for a fixed value of $k$ we could adopt the multiple LP framework of Conitzer and Sandholm~\cite{conitzer2006computing} to include $k$ as an optimization parameter. The idea is simple. Because the range of possible values of $k$ is small $(k \leq C_{max})$ we can simply solve Optimization Goal \ref{goal:MinAdvSuccessAsLP} separately for each value of $k$ and take the best solution --- the one with the smallest value of $\mathcal{P}_{Adv,B}$.

\subsection{Practical CASH Optimization} \label{sec:PracticalAlgorithm}
Theorem \ref{thm:Polytime} states that Optimization Goal \ref{goal:MinAdvSuccessAsLP} can be solved in polynomial time using the ellipsoid algorithm~\cite{khachiyan1980polynomial}. While this is nice in theory the ellipsoid algorithm is rarely deployed in practice because the running time tends to be very large. In this section we develop a heuristic algorithm (Algorithm \ref{alg:HeuristicAlg}) to solve Goal \ref{goal:MinAdvSuccessAsLP} using our separation oracle. While algorithm \ref{alg:HeuristicAlg} is guaranteed to always find the (approximately) optimal solution to Optimization Goal \ref{goal:MinAdvSuccessAsLP}, we do not have any theoretical proof that it will converge to find the optimal solution in polynomial time. However, in all of our experiments we found that Algorithm \ref{alg:HeuristicAlg} converged reasonably quickly. 

The basic idea behind our heuristic algorithm is to start by ignoring all of the type (5) constraints from Goal \ref{goal:MinAdvSuccessAsLP}. We then run a standard LP solver to find the optimal solution to the resulting LP. Finally, we run our separation oracle to determine if this solution violates any type (5) constraints. If it does not then we are done. If the separation oracle does find a violated type (5) constraint then we add this constraint to our LP and solve again. We repeat this process until we have a solution that satisfies all type (5) constraints. Observe that this process must terminate because we will eventually run out of type (5) constraints to add. The hope is that our algorithm will converge much more quickly. In practice, we find that it does  (e.g., at most $25$ iterations).

\paragraph{Further Optimizations} Our separation oracle runs in time $O\left(mn \log mn\right)$ because we sort a list of $mn$ tuples $(pwd,t)$. In practice, the number of passwords $n$ might be very large (e.g., the RockYou dataset contains $n\approx 14.3\times 10^6$ unique passwords).  Fortunately, it is often possible to drastically reduce the time and space requirements of our separation oracle by grouping passwords into equivalence classes. In particular, we group two passwords $pwd_i$ and $pwd_{j}$ into an equivalence class if and only if $p_i = p_j$. This approach reduces running time of our separation oracle to $O\left( mn' \log mn'\right)$, where $n'$ is the number of equivalence classes\footnote{To save computation one could also group passwords into equivalence classes with {\em approximately} equal probabilities, but this representation loses some accuracy and was unnecessary in all of our experiments. }. For example, the RockYou database contains over $10^7$ unique passwords, but we only get $n' = 2040$ equivalence classes.

We can represent our empirical distribution over passwords as a sequence of $n'$ pairs $\left(p_1,n_1\right), \ldots, \left( p_{n'},n_{n'}\right)$. Here, $p_i$ denotes the probability of a password in equivalence class $i$ and $n_i \in \mathbb{N}$ denotes the total number of passwords in equivalence class $i$. We have $\sum_{i=1}^{n'} n_i = n$ and $\sum_{i=1}^{n'} n_i \cdot p_i = 1$. As before we assume that $p_i \geq p_{i+1}$. In most password datasets $n_{n'}$ is the number of passwords that were selected by only one user (e.g., for the RockYou dataset $n_{n'} \approx 11.9\times 10^6$). 

We now argue that this change in view does not fundamentally alter our linear program (Optimization Goal \ref{goal:MinAdvSuccessAsLP}) or our separation oracle. Constraints (1)--(4) in our LP remain unchanged. We need to make a few notational changes to type (5) constraints to ensure that $\mathcal{P}_{Adv,B}$ is at least as large at the success rate of the optimal adversary. We use \[\mathcal{F}_B = \left\{ \left(b_1,\ldots,b_{n'} \right) \in \mathbb{N}^{n'} ~\vline \sum_{i=1}^{n'} b_i \leq B \wedge \forall i. b_i \leq m\cdot n_i \right\}  \ , \]
to describe the space of feasible guessing strategies for an adversary with a threshold $B$. Here, $b_i$ denotes the total number of times the adversary evaluates $\mathbf{H}^k$ while attacking passwords in equivalence class $i \leq n'$. Thus, the range of $b_i$ is $0 \leq b_i \leq m \cdot n_i $  because there are $n_i$ passwords in the equivalence class to attack and he can choose to evaluate  $\mathbf{H}^k$ up to $m$ times for each password. 

Given values $\tilde{p}_1,\ldots,\tilde{p}_m$ and a feasible allocation $b_1,\ldots,b_{n'} \in \mathcal{F}_B$ the probability that adversary will crack the password is at most \[\sum_{i=1}^{n'} p_i  \left( \left(b_i \mod{n_i}\right)\tilde{p}_{\left\lceil \frac{b_i}{n_i} \right\rceil} + \sum_{j=1}^{\left\lfloor \frac{b_i}{n_i} \right\rfloor} n_i \tilde{p}_j \right) \ .\] 
Intuitively, the optimal adversary will spend equal effort ($b_i/n_i$) cracking each password in an equivalence class because they all have the same probability. The $\left(b_i \mod{n_i}\right)$ and $\lfloor b_i/n_i \rfloor$ terms handle the technicality that $b_i$ may not be divisible by $n_i$. Thus, we can replace our type (5) constraints with the constraint \[ \mathcal{P}_{Adv,B} \geq \sum_{i=1}^{n'} p_i  \left( \left(b_i \mod{n_i}\right)\tilde{p}_{\left\lceil \frac{b_i}{n_i} \right\rceil} + \sum_{j=1}^{\left\lfloor \frac{b_i}{n_i} \right\rfloor} n_i \tilde{p}_j \right) \ , \]
for every $\left(b_1,\ldots,b_{n'} \right)\in \mathcal{F}_B$.

Our modified separation oracle works in essentially the same way. We sort the tuples $(i,j)$ using the values $p_{i,j}' = p_i \cdot \tilde{p}_j$ and select the $B$ largest tuples. The only difference is that the adversary is now allowed to select the tuple $(i,j)$ up to $n_i$ times. In this section we will use $\mathbf{SeparationOracle}$ to refer to the modified separation oracle, which runs in time $O\left( mn' \log mn'\right)$ using our compact representation of the empirical password distribution.

Our heuristic algorithm (Algorithm \ref{alg:HeuristicAlg})  takes as input an approximation parameter $\epsilon$. It is allowed to output a solution $\tilde{p}_1,\ldots,\tilde{p}_m, \mathcal{P}_{Adv,B}$ as long as the solution is within $\epsilon$ of optimal --- for any other feasible solution $\tilde{p}_1',\ldots,\tilde{p}_m', \mathcal{P}_{Adv,B}'$ we have $\mathcal{P}_{Adv,B} \leq \mathcal{P}_{Adv,B}'+\epsilon$. We use $\mathbf{Slack}$ to denote a function that computes how badly a linear inequality $C$ is violated. For example, if $C$ denotes the inequality $x+y \geq 2.5$ and we have set $x' = y' = 1$ then $\mathbf{Slack}\left(C,x',y'\right) = 0.5$ (e.g., if we introduced a slack variable $z$ then we would need to select $z'$ such that $\left|z'\right| = 0.5$ to satisfy the inequality $x'+y'+z'\geq 2.5$). 

\begin{algorithm}[H]
 \caption{$\mathbf{Optimize}\left(p,n,B,C_{max},\alpha, \epsilon,m, S\right)$}
 \label{alg:HeuristicAlg}
\begin{algorithmic}[1]
 \renewcommand{\algorithmicrequire}{\textbf{Input:}}
 \renewcommand{\algorithmicensure}{\textbf{Output:}}
\REQUIRE $p_1,\ldots, p_{n'}$, $n_1,\ldots,n_{n'}$,  $B$, $C_{max}$, $\alpha$,  $\epsilon$, $m$,  $S = \left\{k_0,k_1,\ldots, ,k_\tau \right\}$,

\STATE $bestSolution \gets \emptyset$, $bestK \gets k_0$
\STATE $bestSuccessRate \gets 1.0$,  $slack \gets \epsilon$
\FOR{$j=0,\ldots, \tau$}
\STATE $k \gets k_j$
\STATE $C \gets \mathbf{InitialConstraints}(C_{max},\alpha,k)$ \newline 
\COMMENT{Initially, $C$ only includes constraints (1)--(4) \newline in goal \ref{goal:MinAdvSuccessAsLP}}

\STATE $Goal \gets \{ \min \mathcal{P}_{Adv,B} \}$
\STATE $Vrbls \gets \{ \mathcal{P}_{Adv,B}, \tilde{p}_1,\ldots, \tilde{p}_m\}$
\STATE  $\mathcal{P}_{Adv,B}',\tilde{p}_1',\ldots,\tilde{p}_m' \gets  \mathbf{LPSolve}\left(Goal,Vrbls, C\right)$
\STATE $\tilde{p}' \gets \left(\tilde{p}_1',\ldots,\tilde{p}_m' \right)$
\STATE $\vec{p} \gets \left( p_1,\ldots,p_{n'} \right)$
\STATE $\vec{n} \gets \left(  n_1,\ldots,n_{n'}\right)$
\STATE $Sep_{in} \gets \left(\vec{p},\vec{n},\tilde{p}',B,k,C_{SRV,\alpha}, \mathcal{P}_{Adv,B}' \right)$
\STATE $C' \gets \mathbf{SeparationOracle}\left(Sep_{in}\right)$
\WHILE{$\left|\mathbf{Slack}\left(C',\tilde{p}, \mathcal{P}_{Adv,B}'   \right) \right| > \epsilon \wedge \left(C' \neq  \mbox{``Ok"} \right)$ }
\STATE $C \gets C \bigcup \{C'\}$
\STATE  $\mathcal{P}_{Adv,B}',\tilde{p}'\gets\mathbf{LPSolve}\left(Goal,Vrbls, C\right)$ \STATE~~~~~~~~~~~~~~~~~~~~~\COMMENT{$\tilde{p}' = \left(\tilde{p}_1',\ldots,\tilde{p}_m' \right)$}
\STATE $Sep_{in} \gets \left(\vec{p},\vec{n},\tilde{p}',B,k,C_{SRV,\alpha}, \mathcal{P}_{Adv,B}' \right)$
\STATE $C' \gets \mathbf{SeparationOracle}\left(Sep_{in}\right)$
\ENDWHILE
\IF{$bestSuccessRate \geq \mathcal{P}_{Adv,B}'$}
\STATE $bestSolution \gets \tilde{p}_1,\ldots, \tilde{p}_m$
\STATE $bestSuccessRate \gets \mathcal{P}_{Adv,B}'$
\STATE $\left( bestM,bestK\right) \gets \left( m_i,k_i\right)$
\STATE $slack \gets \mathbf{Slack}\left(C',\tilde{p}, \mathcal{P}_{Adv,B}'   \right)$
\ENDIF
\ENDFOR
\RETURN $\tilde{p}_1,\ldots, \tilde{p}_m, bestK$
\end{algorithmic}
\end{algorithm}

\subsection{Choosing a CASH Distribution}
While Algorithm \ref{alg:HeuristicAlg} efficiently solves optimization Goal \ref{goal:MinAdvSuccessAsLP}, it may not yield the optimal distribution for our original Stackelberg game. In particular, while Algorithm \ref{alg:HeuristicAlg} gives the optimal distribution against a threshold-$B$ adversary, the rational adversary might choose to use a different threshold $B^* \neq B$. 

We introduce a heuristic algorithm to find good Stackelberg strategies (CASH distributions) for the defender.  Algorithm \ref{alg:HeuristicAlg2} uses Algorithm \ref{alg:HeuristicAlg} as a subroutine to search for good CASH distributions. Algorithm \ref{alg:HeuristicAlg2} takes as input an (estimate) $\hat{v}$ of the adversary's value and a set $\mathcal{B}$ of potential adversary thresholds $B$ and runs Algorithm \ref{alg:HeuristicAlg} to compute the optimal distribution for each threshold. We then compute the rational value $\hat{v}$ adversary's best response to each of distributions and find the best distribution for the authentication server --- the one which results in the lowest fraction of cracked passwords under the corresponding best adversary  response. Algorithm \ref{alg:HeuristicAlg2} assumes a subroutine $\mathbf{RationalAdvSuccess}\left(p,n,\hat{v},\tilde{p},k\right)$, which computes the fraction of cracked passwords under a value $\hat{v}$ adverary's best response to the CASH distribution $\tilde{p}$ with empirical password distribution defined by the pair $(p,n)$ and a hash cost parameter $k$. 

\begin{algorithm}[H]
 \caption{$\mathbf{FindCASHDistribution}$}
 \label{alg:HeuristicAlg2}
\begin{algorithmic}[1]
 \renewcommand{\algorithmicrequire}{\textbf{Input:}}
 \renewcommand{\algorithmicensure}{\textbf{Output:}}
\REQUIRE $p_1,\ldots, p_{n'}$, $n_1,\ldots,n_{n'}$,  $\hat{v}$, $C_{max}$, $\alpha$,  $\epsilon$, $m$, $S = \left\{k_0,k_1,\ldots,k_\tau \right\}$, $\mathcal{B} = \{B_0,B_1,\ldots,B_\ell\}$

\STATE $\tilde{p}_1,\ldots,\tilde{p}_m \gets 1/m$
\STATE $$k \gets  \frac{C_{max}}{(1-\alpha)m + \alpha \left(\frac{m+1}{2} \right)}$$
\STATE $advSuccess \gets \PAdvUnif{\hat{v}}{C_{max}}$

\FOR{$x=0,\ldots, \ell$}
\STATE $B \gets B_x$
\STATE $\tilde{p}_x, k_x \gets \mathbf{Optimize}\left(p,n,B,C_{max},\alpha,\epsilon,m, S\right)$
\STATE $CS \gets \mathbf{RationalAdvSuccess}\left(p,n,\hat{v},\tilde{p}_x,k_x\right)$
\IF{$CS\leq advSuccess$}
\STATE $\tilde{p} \gets \tilde{p}_x$
\STATE $k \gets k_x$
\STATE $advSuccess \gets CS$
\ENDIF
\ENDFOR
\RETURN $\tilde{p}$, $k$
\end{algorithmic}

\end{algorithm}

We remark that the subroutine $\mathbf{RationalAdvSuccess}$ can be computed in time $O\big(n'm \log mn'\big)$ --- the most expensive step is sorting the $mn'$ pairs $(p_i, \tilde{p}_j)$ based on the value $p_i\cdot \tilde{p}_j$. Once we have these pairs in sorted order there is a simple formula for computing the marginal benefit/costs of a larger threshold $B$. \FullVersion{See the full version\cite{fullVersion} of the paper for more details.}{ See Algorithm \ref{alg:AdvSuccess} in the appendix for more details. }

We remark that Algorithm \ref{alg:HeuristicAlg2} is not guaranteed to always find the optimal solution to optimization goal \ref{goal:MinAdvSuccess}. It may be viewed as a heuristic algorithm that generates many promising candidate CASH distributions and then selects the best distribution among them.

\section{Experimental Results} \label{sec:Evaluation}
In this section we empirically demonstrate that our CASH mechanism can be used to significantly reduce the fraction of accounts that an offline adversary could compromise. We implemented Algorithm \ref{alg:HeuristicAlg2} in C\# using Gurobi as our LP solver, and analyzed CASH using two real-world password distributions $p_1,\ldots,p_n$. The first distribution is based on data from the RockYou password breach ($32+$ million passwords) and the second is based on password frequency data from Yahoo! users (representing $\approx 70$ million passwords). The later dataset was not the result of a security breach. Instead, Yahoo! gave Bonneau~\cite{bonneau2012science} permission to collect and analyze password frequency data in a carefully controlled environment. Yahoo! recently allowed Blocki et al.~\cite{blocki2016differentially} to use a differentially private~\cite{dwork2006calibrating} algorithm to publish this data. Thus, the password frequency data in this data set has been perturbed slightly.  Blocki et al.~\cite{blocki2016differentially} also showed that with high probability the L1 error introduced by their algorithm would be minimal.

In each of our experiments we fix the password correctness rate $\alpha \in \{1,0.95,0.9\}$ and the maximum amortized server cost $C_{max}$ before using Algorithm \ref{alg:HeuristicAlg2} to find a CASH parameters $\tilde{p}_1,\ldots, \tilde{p}_m$ and $k$ subject to the appropriate constraints on the amortized server costs. 

We compare the $\%$ of cracked passwords under three different scenarios: 
\begin{itemize}
\item (Deterministic Key-Stretching) The authentication server selects a hash function $\Hash^k$  with cost parameter  $k = C_{max}$ (achieved through traditional deterministic key-stretching techniques). The rational value $v$ adversary will crack each password with probability $\PAdvDet{v}{k}$ (eq \ref{eq:PAdvDet}).   
\item (Uniform-CASH) The authentication server uses CASH with the uniform distribution. He sets $k$ according to eq \ref{eq:kUnifCash}
to ensure that his amortized costs are at most $C_{max}$. A rational value $v$ adversary will crack each password with probability $\PAdvUnif{v}{C_{max}}$. 

\item (CASH) Given an estimate $\hat{v}$ of the adversary's budget we used Algorithm \ref{alg:HeuristicAlg2} to optimize the CASH parameters $k$ and $\tilde{p}_1,\ldots, \tilde{p}_m$ subject to the constraint that the amortized server cost is at most $C_{max}$ when users enter the wrong password with probability $1-\alpha$. We fixed the parameters $m=50, \epsilon = 0.02$, and we set $\mathcal{B} = \{5\cdot C_{max} \times 10^4$, $C_{max}\times 10^6$, $C_{max}\times 10^7$, $1.5\cdot C_{max}\times 10^7$, $2.0\cdot C_{max} \times 10^7$,  $2.5\cdot C_{max} \times 10^7$, $2.65\cdot C_{max} \times 10^7$, $2.8\cdot C_{max} \times 10^7$, $3.0\cdot C_{max} \times 10^7$, $5.0\cdot C_{max} \times 10^7$, $C_{max} \times 10^8\}$. Thus, Algorithm \ref{alg:HeuristicAlg2} computes the optimal distribution against a threshold $B$ adversary for each $B \in \mathcal{B}$, and selects the best distribution $\tilde{p}$ against a value $\hat{v}$ adversary.  $\PAdvCASH{\hat{v}}{\hat{v}}{C_{max}}$ will denote the fraction of cracked passwords when the true value is $v=\hat{v}$. When the adversary's true value is $v \neq \hat{v}$, $\PAdvCASH{v}{\hat{v}}{C_{max}}$  will denote the fraction of cracked passwords.  
\end{itemize}

Our results indicate that an authentication server could significantly reduce the fraction of compromised passwords in an offline attack by adopting our optimal CASH mechanism instead of deterministic key-stretching or uniform-CASH. These results held robustly for both the RockYou and Yahoo! password distributions.


\subsubsection{Password Datasets} \label{subsubsec:RockYou} We use two password frequency datasets, RockYou and Yahoo!, to analyze our CASH mechanism. The RockYou dataset contains passwords from $N\approx 32.6$ million RockYou users, and the Yahoo! dataset contains data from $N\approx 70$ million Yahoo! users. We used frequency data from each of these datasets to obtain an empirical password distribution $p_1 \geq p_2 \geq p_3 \ldots \geq p_n$ over $\PasswordSpace$. 

The RockYou dataset is based on actual user passwords which were leaked during the infamous RockYou security breach (RockYou had been storing these passwords in the clear). The total number of unique passwords in the dataset was $n \approx 14.3$ million. Approximately, $11.9$ million of these passwords were unique to one RockYou user. The other $\approx 2.5$ million passwords were used by multiple users. The most popular password ($pwd_1=$ `123456') was shared by $\approx 0.3$ million RockYou users ($p_1 \approx 0.01$). RockYou did not impose strict password restrictions on its users (e.g. users were allowed to select passwords consisting of only lowercase letters or only numbers). 

We also used (perturbed) password frequency data from a dataset of $N\approx 70$ million Yahoo! passwords. See~\cite{bonneau2012science} for more details about how this data was collected and see~\cite{blocki2016differentially} for more details about how the frequency data was perturbed to satisfy the rigorous notion of differentially privacy~\cite{dwork2006calibrating}. Blocki et al.~\cite{blocki2016differentially} proved that with high probability the L1 distortion of the perturbed frequency data is bounded by $O\left( \sqrt{N}/\epsilon\right)$, where the privacy parameter was set to $\epsilon=0.25$ when the Yahoo! dataset was published. Thus, the perturbed dataset will also still give us a good estimate of the empirical password distribution.

\subsection{Results}
Our first set of experimental results are shown in Figures \ref{fig:YahooResults1} and \ref{fig:RockYouResults1}. These plots were computed under the assumption that $\alpha = 1$ (users always enter their passwords correctly), and that $\hat{v}=v$ (the defender knows the exact adversary value). The results show that for some (higher) adversary values our non-uniform CASH distributions improves significantly on the cost-equivalent versions of uniform CASH ($50\%$ reduction in cracked passwords) and deterministic key-stretching ($56\%$ reduction in cracked passwords).\footnote{We note that we would expect to see relatively high adversary  values $v/C_{max}$ in the offline setting because $C_{max}$ will typically be quite small (e.g., $\$10^{-6}$). }  Figures \ref{fig:YahooResults90} and \ref{fig:RockYouResults90} (resp. Figures \ref{fig:YahooResults95} and \ref{fig:RockYouResults95}) show the same results under the assumptions that $\alpha = 0.9$ (resp. $\alpha=0.95$). 

Figures \ref{fig:YahooRobust2Results95} and \ref{fig:YahooRobust1Results95} (resp. Figures \ref{fig:RockYouRobust1Results95} and \ref{fig:RockYouRobust2Results95}) explore the effect of a wrong estimate $\hat{v} \neq v$ of the adversary's value for both the RockYou and Yahoo! datasets. Despite receiving the wrong estimate $\hat{v}$ our algorithm returns a distribution that is (almost always) slightly better than the corresponding uniform CASH distribution. Both distributions still significantly outperform the cost equivalent deterministic key-stretching solution. 

\FullVersion{The full version~\cite{fullVersion} of this paper contains additional plots exploring what happens when the defender uses the wrong empirical password distribution when searching for a good CASH distribution $\tilde{p}$ (e.g., if the defender optimizes $\tilde{p}$ under the assumption that the empirical password distribution is given by the Yahoo! dataset when the actual distribution is given by the RockYou dataset).
}{Figures \ref{fig:RockYouResultsOptForYahoo95} and  \ref{fig:YahooResultsOptForRockYou95}  in the appendix explore what happens when the defender uses the wrong empirical password distribution when searching for a good CASH distribution $\tilde{p}$ (e.g., if the defender optimizes $\tilde{p}$ under the assumption that the empirical password distribution is given by the Yahoo! dataset when the actual distribution is given by the RockYou dataset). } Briefly, these plots show that non-uniform CASH significantly outperforms deterministic key-stretching even when non-uniform CASH is optimized under the wrong distribution and non-uniform CASH slightly outperforms uniform CASH on most, but not all, of the curve.

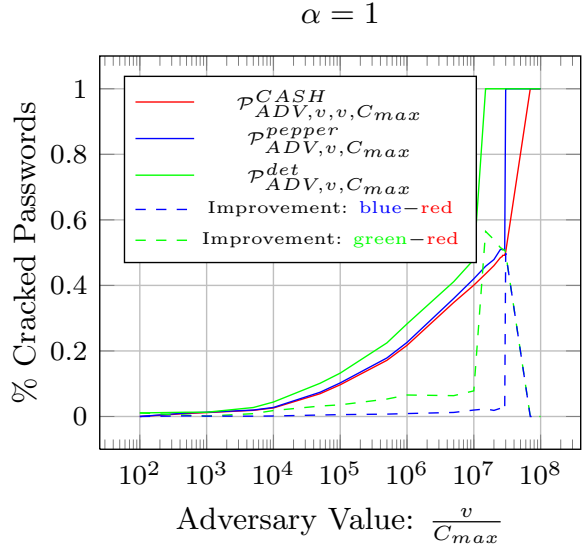
\begin{figure}[!t]
\centering
    \begin{tikzpicture}[scale=1.3]  
   \begin{semilogxaxis}[
    title style={align=center},
    title={ {\small $\alpha=1$ }},
    xlabel={Adversary Value: $\frac{v}{C_{max}}$},
    ylabel={$\%$ Cracked Passwords},
    ylabel shift = -3pt,
    grid=major,
    small,
    cycle list = {{red, mark=none}, {blue, mark=none}, {green, mark=none},{green, dashed, mark=none}, {blue, dashed, mark=none},  {blue, dashed, mark=none}, {red, dashed, mark=none},{brown, dashed, mark=none} },
    legend style = {font=\tiny, at={(.05,.95)}, anchor=north west},
    legend entries = { $\PAdvCASH{v}{v}{C_{max}}$, $\PAdvUnif{v}{C_{max}}$, $\PAdvDet{v}{C_{max}}$, Improvement: \textcolor{blue}{blue}$-$\textcolor{red}{red}, Improvement: \textcolor{green}{green}$-$\textcolor{red}{red}}
   ] 

    \addlegendimage{no markers, red, mark=square*}
    \addlegendimage{no markers, blue, mark=square*}
    \addlegendimage{no markers, green, mark=square*}
    \addlegendimage{no markers, dashed, blue}
    \addlegendimage{no markers, dashed, green}
\addplot coordinates {  (100,0)  (500,0.00901095679831186)  (1000,0.0118882075379984)  (5000,0.0187377670340035)  (10000,0.0260329771951511)  (50000,0.0695289722433699)  (100000,0.0967762178707898)  (500000,0.171395066946001)  (1000000,0.217079290083012)  (5000000,0.347777554099186)  (10000000,0.400698912520617)  (15000000,0.435241555034486)  (20000000,0.460298959755672)  (25000000,0.484752773595168)  (26500000,0.487934163505628)  (27000000,0.491586385753338)  (27500000,0.491586385753338)  (28000000,0.491891527286909)  (29000000,0.491904229037079)  (30000000,0.501829516819533)  (70000000,0.999999999999975)  (100000000,0.999999999999975)  };
\addplot coordinates {  (100,0)  (500,0.0108687224894377)  (1000,0.0130192581998815)  (5000,0.0196877875530742)  (10000,0.0275769138479969)  (50000,0.0740168548263362)  (100000,0.102497142299001)  (500000,0.178594664053884)  (1000000,0.225985726653441)  (5000000,0.360129776428412)  (10000000,0.420317345392629)  (15000000,0.45834913690049)  (20000000,0.478853387778074)  (25000000,0.50975028086399)  (26500000,0.50975028086399)  (27000000,0.50975028086399)  (27500000,0.50975028086399)  (28000000,0.50975028086399)  (29000000,0.50975028086399)  (30000000,1)  (70000000,1)  (100000000,1)  };
 \addplot coordinates {  (100,0.0108687224894377)  (500,0.0130192581998815)  (1000,0.0130192581998815)  (5000,0.0273838584095427)  (10000,0.0441510096695537)  (50000,0.101412040578669)  (100000,0.132372669808665)  (500000,0.224601121331902)  (1000000,0.282624504055384)  (5000000,0.411368138539665)  (10000000,0.478853387778074)  (15000000,1)  (20000000,1)  (25000000,1)  (26500000,1)  (27000000,1)  (27500000,1)  (28000000,1)  (29000000,1)  (30000000,1)  (70000000,1)  (100000000,1)  };
\addplot coordinates {  (100,0.0108687224894377)  (500,0.00400830140156962)  (1000,0.0011310506618831)  (5000,0.00864609137553916)  (10000,0.0181180324744026)  (50000,0.0318830683352988)  (100000,0.0355964519378749)  (500000,0.0532060543859009)  (1000000,0.0655452139723718)  (5000000,0.0635905844404786)  (10000000,0.0781544752574576)  (15000000,0.564758444965514)  (20000000,0.539701040244328)  (25000000,0.515247226404832)  (26500000,0.512065836494372)  (27000000,0.508413614246662)  (27500000,0.508413614246662)  (28000000,0.508108472713091)  (29000000,0.508095770962921)  (30000000,0.498170483180467)  (70000000,2.48689957516035E-14)  (100000000,2.48689957516035E-14)  };
\addplot coordinates {  (100,0)  (500,0.00185776569112582)  (1000,0.0011310506618831)  (5000,0.000950020519070692)  (10000,0.00154393665284583)  (50000,0.00448788258296633)  (100000,0.00572092442821077)  (500000,0.0071995971078832)  (1000000,0.00890643657042919)  (5000000,0.0123522223292253)  (10000000,0.0196184328720127)  (15000000,0.0231075818660044)  (20000000,0.0185544280224019)  (25000000,0.0249975072688218)  (26500000,0.0218161173583619)  (27000000,0.0181638951106518)  (27500000,0.0181638951106518)  (28000000,0.0178587535770803)  (29000000,0.0178460518269108)  (30000000,0.498170483180467)  (70000000,2.48689957516035E-14)  (100000000,2.48689957516035E-14)  };

   \end{semilogxaxis} 
  \end{tikzpicture}
  
  \caption{Yahoo Dataset: $\alpha = 1$.}
\label{fig:YahooResults1}
\end{figure}
 
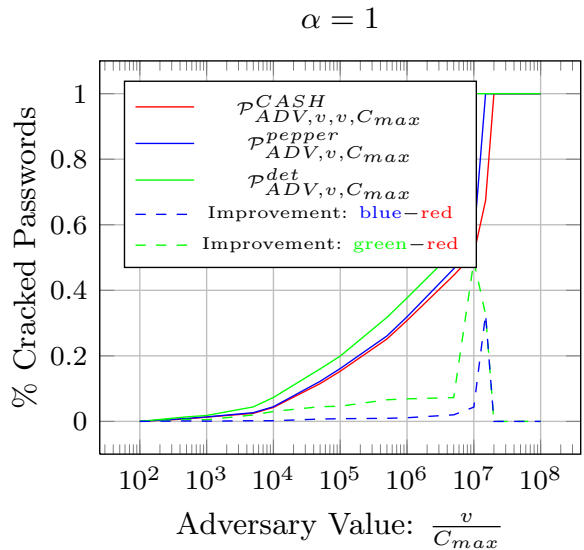
\begin{figure}[!t]
\centering
    \begin{tikzpicture}[scale=1.3]  
   \begin{semilogxaxis}[
    title style={align=center},
    title={ {\small $\alpha=1$ }},
    xlabel={Adversary Value: $\frac{v}{C_{max}}$},
    xlabel={Adversary Value: $\frac{v}{C_{max}}$},
    ylabel={$\%$ Cracked Passwords},
    ylabel shift = -3pt,
    grid=major,
    small,
    cycle list = {{red, mark=none}, {blue, mark=none}, {green, mark=none},{green, dashed, mark=none}, {blue, dashed, mark=none},  {blue, dashed, mark=none}, {red, dashed, mark=none},{brown, dashed, mark=none} },
    legend style = {font=\tiny, at={(.05,.95)}, anchor=north west},
    legend entries = { $\PAdvCASH{v}{v}{C_{max}}$, $\PAdvUnif{v}{C_{max}}$, $\PAdvDet{v}{C_{max}}$, Improvement: \textcolor{blue}{blue}$-$\textcolor{red}{red}, Improvement: \textcolor{green}{green}$-$\textcolor{red}{red}}
   ] 

    \addlegendimage{no markers, red, mark=square*}
    \addlegendimage{no markers, blue, mark=square*}
    \addlegendimage{no markers, green, mark=square*}
    \addlegendimage{no markers, dashed, blue}
    \addlegendimage{no markers, dashed, green}
\addplot coordinates {  (100,0)  (500,0.0076046384413692)  (1000,0.0130127814866306)  (5000,0.0245423032689714)  (10000,0.0427309815113077)  (50000,0.11520218186584)  (100000,0.15297695370972)  (500000,0.251342450785729)  (1000000,0.307887790562047)  (5000000,0.445734691700375)  (10000000,0.510814832636074)  (15000000,0.673482337523552)  (20000000,0.999999999999993)  (25000000,0.999999999999993)  (26500000,0.999999999999993)  (27000000,0.999999999999993)  (27500000,0.999999999999993)  (28000000,0.999999999999993)  (29000000,0.999999999999993)  (30000000,0.999999999999993)  (70000000,0.999999999999993)  (100000000,0.999999999999993)  };
\addplot coordinates {  (100,0)  (500,0.00891714075850031)  (1000,0.0136977788934083)  (5000,0.0265629142590948)  (10000,0.0447626485934529)  (50000,0.12227710813367)  (100000,0.161237660331497)  (500000,0.260629600825534)  (1000000,0.319035371416001)  (5000000,0.465731506185799)  (10000000,0.55428368364662)  (15000000,1)  (20000000,1)  (25000000,1)  (26500000,1)  (27000000,1)  (27500000,1)  (28000000,1)  (29000000,1)  (30000000,1)  (70000000,1)  (100000000,1)  };
 \addplot coordinates {  (100,0)  (500,0.0136977788934083)  (1000,0.0180747779954648)  (5000,0.0441965417827129)  (10000,0.0730608733055595)  (50000,0.160089926850547)  (100000,0.199287847017616)  (500000,0.317302545367371)  (1000000,0.376634078642379)  (5000000,0.518411276766697)  (10000000,1)  (15000000,1)  (20000000,1)  (25000000,1)  (26500000,1)  (27000000,1)  (27500000,1)  (28000000,1)  (29000000,1)  (30000000,1)  (70000000,1)  (100000000,1)  };
\addplot coordinates {  (100,0)  (500,0.00609314045203905)  (1000,0.00506199650883418)  (5000,0.0196542385137415)  (10000,0.0303298917942517)  (50000,0.0448877449847074)  (100000,0.0463108933078961)  (500000,0.065960094581642)  (1000000,0.0687462880803322)  (5000000,0.0726765850663217)  (10000000,0.489185167363926)  (15000000,0.326517662476448)  (20000000,7.105427357601E-15)  (25000000,7.105427357601E-15)  (26500000,7.105427357601E-15)  (27000000,7.105427357601E-15)  (27500000,7.105427357601E-15)  (28000000,7.105427357601E-15)  (29000000,7.105427357601E-15)  (30000000,7.105427357601E-15)  (70000000,7.105427357601E-15)  (100000000,7.105427357601E-15)  };
\addplot coordinates {  (100,0)  (500,0.00131250231713111)  (1000,0.000684997406777665)  (5000,0.00202061099012341)  (10000,0.00203166708214515)  (50000,0.0070749262678306)  (100000,0.00826070662177641)  (500000,0.00928715003980463)  (1000000,0.0111475808539544)  (5000000,0.0199968144854238)  (10000000,0.0434688510105455)  (15000000,0.326517662476448)  (20000000,7.105427357601E-15)  (25000000,7.105427357601E-15)  (26500000,7.105427357601E-15)  (27000000,7.105427357601E-15)  (27500000,7.105427357601E-15)  (28000000,7.105427357601E-15)  (29000000,7.105427357601E-15)  (30000000,7.105427357601E-15)  (70000000,7.105427357601E-15)  (100000000,7.105427357601E-15)  };

   \end{semilogxaxis} 
  \end{tikzpicture}
  
  \caption{RockYou Dataset: $\alpha = 1$.}
\label{fig:RockYouResults1}
\end{figure} 
  
   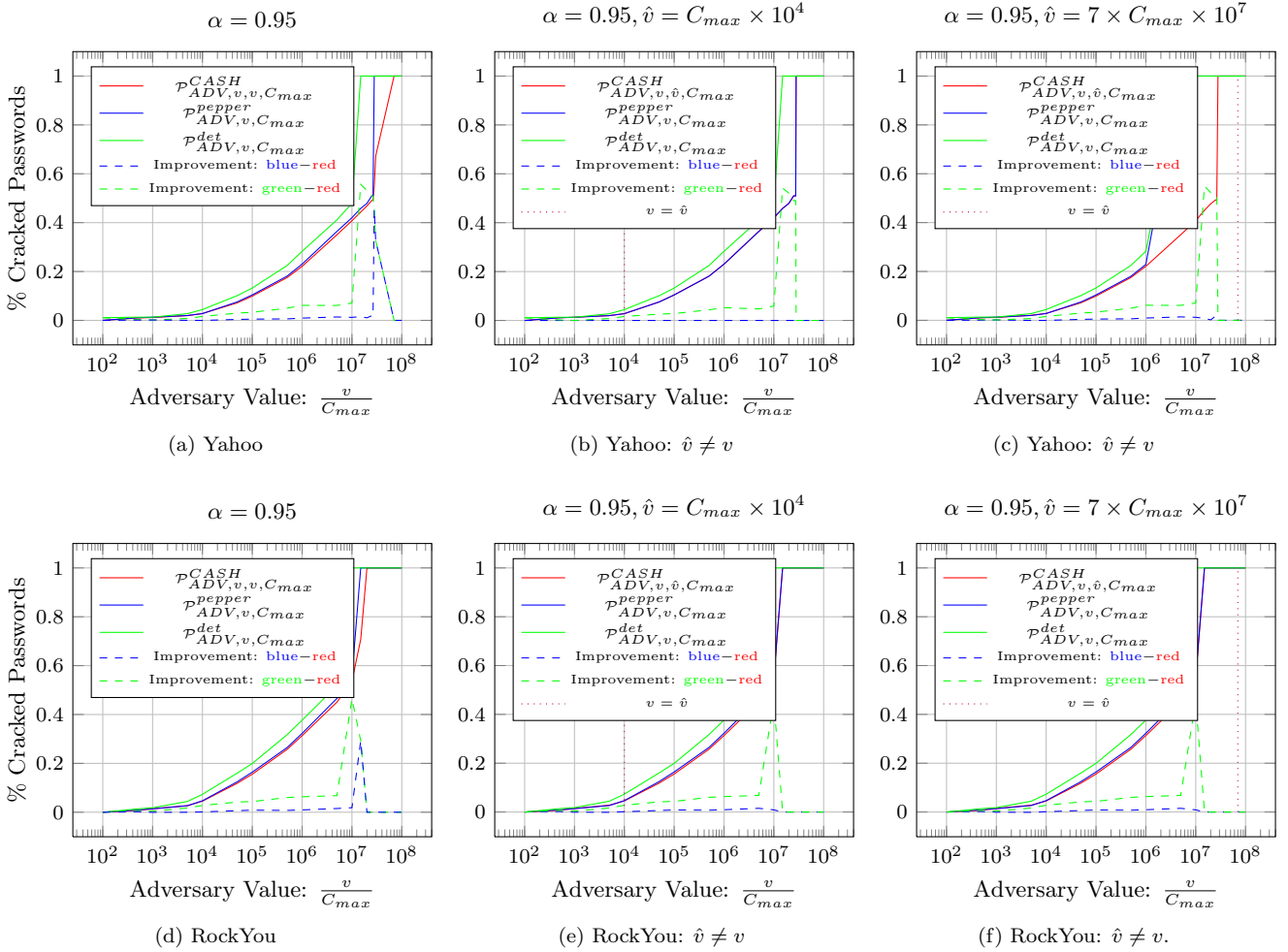
\begin{figure}[!t]
   \centering
\subfloat[Yahoo]{
    \begin{tikzpicture}[scale=1.0]  
   \begin{semilogxaxis}[
    title style={align=center},
    title={ {\small $\alpha=0.95$ }},
    xlabel={Adversary Value: $\frac{v}{C_{max}}$},
    ylabel={$\%$ Cracked Passwords},
    ylabel shift = -3pt,
    grid=major,
    small,
    cycle list = {{red, mark=none}, {blue, mark=none}, {green, mark=none},{green, dashed, mark=none}, {blue, dashed, mark=none},  {blue, dashed, mark=none}, {red, dashed, mark=none},{brown, dashed, mark=none} },
    legend style = {font=\tiny, at={(.05,.95)}, anchor=north west},
    legend entries = { $\PAdvCASH{v}{v}{C_{max}}$, $\PAdvUnif{v}{C_{max}}$, $\PAdvDet{v}{C_{max}}$, Improvement: \textcolor{blue}{blue}$-$\textcolor{red}{red}, Improvement: \textcolor{green}{green}$-$\textcolor{red}{red}}
   ] 

    \addlegendimage{no markers, red, mark=square*}
    \addlegendimage{no markers, blue, mark=square*}
    \addlegendimage{no markers, green, mark=square*}
    \addlegendimage{no markers, dashed, blue}
    \addlegendimage{no markers, dashed, green}
\addplot coordinates {  (100,0)  (500,0.00974816549150774)  (1000,0.0115631568630905)  (5000,0.0194896102497567)  (10000,0.0280671385165211)  (50000,0.071791685890364)  (100000,0.0992383142594701)  (500000,0.175014737387696)  (1000000,0.221098434924485)  (5000000,0.35036457681152)  (10000000,0.407944399818817)  (15000000,0.443322031170548)  (20000000,0.468392975247313)  (25000000,0.490721988368897)  (26500000,0.492844202602487)  (28000000,0.546494585567765)  (30000000,0.673166853909312)  (70000000,0.999999999999971)  (100000000,0.999999999999971)  };
\addplot coordinates {  (100,0)  (500,0.0108687224894377)  (1000,0.0130192581998815)  (5000,0.0196877875530742)  (10000,0.0277625668318636)  (50000,0.0756145584896869)  (100000,0.104286400708258)  (500000,0.181055900840701)  (1000000,0.230322497241287)  (5000000,0.364249437207828)  (10000000,0.420317345392629)  (15000000,0.45834913690049)  (20000000,0.478853387778074)  (25000000,0.50975028086399)  (26500000,0.50975028086399)  (28000000,1)  (30000000,1)  (70000000,1)  (100000000,1)  };
 \addplot coordinates {  (100,0.0108687224894377)  (500,0.0130192581998815)  (1000,0.0130192581998815)  (5000,0.0273838584095427)  (10000,0.0441510096695537)  (50000,0.101412040578669)  (100000,0.132372669808665)  (500000,0.224601121331902)  (1000000,0.282624504055384)  (5000000,0.411368138539665)  (10000000,0.478853387778074)  (15000000,1)  (20000000,1)  (25000000,1)  (26500000,1)  (28000000,1)  (30000000,1)  (70000000,1)  (100000000,1)  };
\addplot coordinates {  (100,0.0108687224894377)  (500,0.00327109270837374)  (1000,0.001456101336791)  (5000,0.00789424815978592)  (10000,0.0160838711530326)  (50000,0.0296203546883047)  (100000,0.0331343555491946)  (500000,0.0495863839442057)  (1000000,0.0615260691308984)  (5000000,0.0610035617281453)  (10000000,0.0709089879592568)  (15000000,0.556677968829452)  (20000000,0.531607024752687)  (25000000,0.509278011631104)  (26500000,0.507155797397513)  (28000000,0.453505414432235)  (30000000,0.326833146090688)  (70000000,2.89768209427166E-14)  (100000000,2.89768209427166E-14)  };
\addplot coordinates {  (100,0)  (500,0.00112055699792995)  (1000,0.001456101336791)  (5000,0.000198177303317452)  (10000,-0.000304571684657504)  (50000,0.00382287259932285)  (100000,0.00504808644878749)  (500000,0.00604116345300501)  (1000000,0.00922406231680187)  (5000000,0.0138848603963081)  (10000000,0.0123729455738119)  (15000000,0.0150271057299424)  (20000000,0.0104604125307615)  (25000000,0.0190282924950931)  (26500000,0.0169060782615023)  (28000000,0.453505414432235)  (30000000,0.326833146090688)  (70000000,2.89768209427166E-14)  (100000000,2.89768209427166E-14)  };

   \end{semilogxaxis} 
  \end{tikzpicture}

\label{fig:YahooResults95}}
 \subfloat[Yahoo: $\hat{v}\neq v$]{   \begin{tikzpicture}[scale=1.0]  
   \begin{semilogxaxis}[
    title style={align=center},
    title={ {\small $\alpha=0.95, \hat{v} = C_{max}\times 10^4$ }},
       xlabel={Adversary Value: $\frac{v}{C_{max}}$},
    ylabel shift = -3pt,
    grid=major,
    small,
    cycle list = {{red, mark=none}, {blue, mark=none}, {green, mark=none},{blue, dashed, mark=none}, {green, dashed, mark=none},  {purple, dotted, mark=none}, {red, dashed, mark=none},{brown, dashed, mark=none} },
    legend style = {font=\tiny, at={(.05,.95)}, anchor=north west},
    legend entries = { $\PAdvCASH{v}{\hat{v}}{C_{max}}$, $\PAdvUnif{v}{C_{max}}$, $\PAdvDet{v}{C_{max}}$, Improvement: \textcolor{blue}{blue}$-$\textcolor{red}{red}, Improvement: \textcolor{green}{green}$-$\textcolor{red}{red}, $v = \hat{v}$  }
   ] 

    \addlegendimage{no markers, red, mark=square*}
    \addlegendimage{no markers, blue, mark=square*}
    \addlegendimage{no markers, green, mark=square*}
    \addlegendimage{no markers, dashed, blue}
    \addlegendimage{no markers, dashed, green}
    \addlegendimage{no markers, dotted, purple}
\addplot coordinates {  (100,0)  (500,0.0108687224894377)  (1000,0.0130192581998815)  (5000,0.0196877875530742)  (10000,0.0277625668318637)  (50000,0.0756145584896891)  (100000,0.104286400708252)  (500000,0.181055900840683)  (1000000,0.230322497241268)  (5000000,0.364249437207804)  (10000000,0.420317345392608)  (15000000,0.458349136900469)  (20000000,0.478853387778053)  (25000000,0.509750280863969)  (26500000,0.509750280863969)  (27000000,0.509750280863969)  (27500000,0.509750280863969)  (28000000,0.999999999999979)  (29000000,0.999999999999979)  (30000000,0.999999999999979)  (70000000,0.999999999999979)  (100000000,0.999999999999979)  };
\addplot coordinates {  (100,0)  (500,0.0108687224894377)  (1000,0.0130192581998815)  (5000,0.0196877875530742)  (10000,0.0277625668318636)  (50000,0.0756145584896869)  (100000,0.104286400708258)  (500000,0.181055900840701)  (1000000,0.230322497241287)  (5000000,0.364249437207828)  (10000000,0.420317345392629)  (15000000,0.45834913690049)  (20000000,0.478853387778074)  (25000000,0.50975028086399)  (26500000,0.50975028086399)  (27000000,0.50975028086399)  (27500000,0.50975028086399)  (28000000,1)  (29000000,1)  (30000000,1)  (70000000,1)  (100000000,1)  };
 \addplot coordinates {  (100,0.0108687224894377)  (500,0.0130192581998815)  (1000,0.0130192581998815)  (5000,0.0273838584095427)  (10000,0.0441510096695537)  (50000,0.101412040578669)  (100000,0.132372669808665)  (500000,0.224601121331902)  (1000000,0.282624504055384)  (5000000,0.411368138539665)  (10000000,0.478853387778074)  (15000000,1)  (20000000,1)  (25000000,1)  (26500000,1)  (27000000,1)  (27500000,1)  (28000000,1)  (29000000,1)  (30000000,1)  (70000000,1)  (100000000,1)  };

\addplot coordinates {  (100,0)  (500,0)  (1000,1.38777878078145E-17)  (5000,-3.46944695195361E-18)  (10000,-1.38777878078145E-16)  (50000,-2.22044604925031E-15)  (100000,5.93969318174459E-15)  (500000,1.77358128183869E-14)  (1000000,1.8846035843012E-14)  (5000000,2.37587727269784E-14)  (10000000,2.17048601314218E-14)  (15000000,2.09832151654155E-14)  (20000000,2.07056594092592E-14)  (25000000,2.1094237467878E-14)  (26500000,2.1094237467878E-14)  (27000000,2.1094237467878E-14)  (27500000,2.1094237467878E-14)  (28000000,2.14273043752655E-14)  (29000000,2.14273043752655E-14)  (30000000,2.14273043752655E-14)  (70000000,2.14273043752655E-14)  (100000000,2.14273043752655E-14)  };
\addplot coordinates {  (100,0.0108687224894377)  (500,0.0021505357104438)  (1000,1.38777878078145E-17)  (5000,0.00769607085646846)  (10000,0.01638844283769)  (50000,0.0257974820889796)  (100000,0.0280862691004131)  (500000,0.0435452204912185)  (1000000,0.0523020068141154)  (5000000,0.0471187013318609)  (10000000,0.0585360423854666)  (15000000,0.541650863099531)  (20000000,0.521146612221947)  (25000000,0.490249719136031)  (26500000,0.490249719136031)  (27000000,0.490249719136031)  (27500000,0.490249719136031)  (28000000,2.14273043752655E-14)  (29000000,2.14273043752655E-14)  (30000000,2.14273043752655E-14)  (70000000,2.14273043752655E-14)  (100000000,2.14273043752655E-14)  };

\addplot coordinates { (10000, 0)  (10000, 0.5)  (10000, 1) };
   \end{semilogxaxis} 
  \end{tikzpicture}  
\label{fig:YahooRobust2Results95}}
\subfloat[Yahoo: $\hat{v}\neq v$]{
    \begin{tikzpicture}[scale=1.0]  
   \begin{semilogxaxis}[
    title style={align=center},
    title={ {\small $\alpha=0.95, \hat{v} = 7\times C_{max} \times 10^7$ }},
       xlabel={Adversary Value: $\frac{v}{C_{max}}$},
    ylabel shift = -3pt,
    grid=major,
    small,
    cycle list = {{red, mark=none}, {blue, mark=none}, {green, mark=none},{blue, dashed, mark=none}, {green, dashed, mark=none},  {purple, dotted, mark=none}, {red, dashed, mark=none},{brown, dashed, mark=none} },
    legend style = {font=\tiny, at={(.05,.95)}, anchor=north west},
    legend entries = { $\PAdvCASH{v}{\hat{v}}{C_{max}}$, $\PAdvUnif{v}{C_{max}}$, $\PAdvDet{v}{C_{max}}$, Improvement: \textcolor{blue}{blue}$-$\textcolor{red}{red}, Improvement: \textcolor{green}{green}$-$\textcolor{red}{red}, $v = \hat{v}$  }
   ] 
   
    \addlegendimage{no markers, red, mark=square*}
    \addlegendimage{no markers, blue, mark=square*}
    \addlegendimage{no markers, green, mark=square*}
    \addlegendimage{no markers, dashed, blue}
    \addlegendimage{no markers, dashed, green}
    \addlegendimage{no markers, dotted, purple}
\addplot coordinates {  (100,0)  (500,0.0102822813124222)  (1000,0.0115631568630905)  (5000,0.0202897976021)  (10000,0.0283029983493156)  (50000,0.071791685890364)  (100000,0.0992383142594701)  (500000,0.175014737387696)  (1000000,0.221098434924485)  (5000000,0.35036457681152)  (10000000,0.407944399818817)  (15000000,0.450602587267483)  (20000000,0.477161630108202)  (25000000,0.492631378777674)  (26500000,0.492844202602487)  (28000000,0.999999999999971)  (30000000,0.999999999999971)  (70000000,0.999999999999971)  (100000000,0.999999999999971)  };

 \addplot coordinates {  (100,0)  (500,0.0108687224894377)  (1000,0.0130192581998815)  (5000,0.0196877875530742)  (10000,0.0277625668318636)  (50000,0.0751253327190499)  (100000,0.103779181056781)  (500000,0.180114995472598)  (1000000,0.228825830589675)  (5000000,1)  (10000000,1)  (15000000,1)  (20000000,1)  (25000000,1)  (26500000,1)  (28000000,1)  (30000000,1)  (70000000,1)  (100000000,1)  };
\addplot coordinates {  (100,0.0108687224894377)  (500,0.0130192581998815)  (1000,0.0130192581998815)  (5000,0.0273838584095427)  (10000,0.0441510096695537)  (50000,0.101412040578669)  (100000,0.132372669808665)  (500000,0.224601121331902)  (1000000,0.282624504055384)  (5000000,1)  (10000000,1)  (15000000,1)  (20000000,1)  (25000000,1)  (26500000,1)  (28000000,1)  (30000000,1)  (70000000,1)  (100000000,1)  };
\addplot coordinates {  (100,0)  (500,0.000586441177015484)  (1000,0.001456101336791)  (5000,-0.000602010049025815)  (10000,-0.000540431517452054)  (50000,0.00382287259932285)  (100000,0.00504808644878749)  (500000,0.00604116345300501)  (1000000,0.00922406231680187)  (5000000,0.0138848603963081)  (10000000,0.0123729455738119)  (15000000,0.00774654963300692)  (20000000,0.00169175766987201)  (25000000,0.017118902086316)  (26500000,0.0169060782615023)  (28000000,2.89768209427166E-14)  (30000000,2.89768209427166E-14)  (70000000,2.89768209427166E-14)  (100000000,2.89768209427166E-14)  };
\addplot coordinates {  (100,0.0108687224894377)  (500,0.00273697688745928)  (1000,0.001456101336791)  (5000,0.00709406080744265)  (10000,0.0158480113202381)  (50000,0.0296203546883047)  (100000,0.0331343555491946)  (500000,0.0495863839442057)  (1000000,0.0615260691308984)  (5000000,0.0610035617281453)  (10000000,0.0709089879592568)  (15000000,0.549397412732517)  (20000000,0.522838369891798)  (25000000,0.507368621222326)  (26500000,0.507155797397513)  (28000000,2.89768209427166E-14)  (30000000,2.89768209427166E-14)  (70000000,2.89768209427166E-14)  (100000000,2.89768209427166E-14)  };

\addplot coordinates { (70000000, 0)  (70000000, 0.5)  (70000000, 1) };
   \end{semilogxaxis} 
  \end{tikzpicture}
  
\label{fig:YahooRobust1Results95}}
 
\subfloat[RockYou]{
    \begin{tikzpicture}[scale=1.0]  
   \begin{semilogxaxis}[
    title style={align=center},
    title={ {\small $\alpha=0.95$ }},
    xlabel={Adversary Value: $\frac{v}{C_{max}}$},
    ylabel={$\%$ Cracked Passwords},
    ylabel shift = -3pt,
    grid=major,
    small,
    cycle list = {{red, mark=none}, {blue, mark=none}, {green, mark=none},{green, dashed, mark=none}, {blue, dashed, mark=none},  {blue, dashed, mark=none}, {red, dashed, mark=none},{brown, dashed, mark=none} },
    legend style = {font=\tiny, at={(.05,.95)}, anchor=north west},
    legend entries = { $\PAdvCASH{v}{v}{C_{max}}$, $\PAdvUnif{v}{C_{max}}$, $\PAdvDet{v}{C_{max}}$, Improvement: \textcolor{blue}{blue}$-$\textcolor{red}{red}, Improvement: \textcolor{green}{green}$-$\textcolor{red}{red}}
   ] 

    \addlegendimage{no markers, red, mark=square*}
    \addlegendimage{no markers, blue, mark=square*}
    \addlegendimage{no markers, green, mark=square*}
    \addlegendimage{no markers, dashed, blue}
    \addlegendimage{no markers, dashed, green}
\addplot coordinates {  (100,0)  (500,0.00703895183690825)  (1000,0.0141148535816532)  (5000,0.0267722121514504)  (10000,0.0454018848661649)  (50000,0.118439836836895)  (100000,0.155703757358036)  (500000,0.257165070416829)  (1000000,0.313403406714585)  (5000000,0.450583879056243)  (10000000,0.535863032524142)  (15000000,0.707445744531801)  (20000000,0.999999999999972)  (25000000,0.999999999999972)  (26500000,0.999999999999972)  (28000000,0.999999999999972)  (30000000,0.999999999999972)  (70000000,0.999999999999972)  (100000000,0.999999999999972)  };
\addplot coordinates {  (100,0)  (500,0.00891714075850031)  (1000,0.0136977788934083)  (5000,0.0265629142590948)  (10000,0.0465878576790853)  (50000,0.124664436714369)  (100000,0.164009427486493)  (500000,0.264563609156202)  (1000000,0.322603129466177)  (5000000,0.465731506185799)  (10000000,0.55428368364662)  (15000000,1)  (20000000,1)  (25000000,1)  (26500000,1)  (28000000,1)  (30000000,1)  (70000000,1)  (100000000,1)  };
 \addplot coordinates {  (100,0)  (500,0.0136977788934083)  (1000,0.0180747779954648)  (5000,0.0441965417827129)  (10000,0.0730608733055595)  (50000,0.160089926850547)  (100000,0.199287847017616)  (500000,0.317302545367371)  (1000000,0.376634078642379)  (5000000,0.518411276766697)  (10000000,1)  (15000000,1)  (20000000,1)  (25000000,1)  (26500000,1)  (28000000,1)  (30000000,1)  (70000000,1)  (100000000,1)  };
\addplot coordinates {  (100,0)  (500,0.0066588270565)  (1000,0.0039599244138116)  (5000,0.0174243296312625)  (10000,0.0276589884393946)  (50000,0.0416500900136517)  (100000,0.0435840896595805)  (500000,0.0601374749505419)  (1000000,0.0632306719277944)  (5000000,0.0678273977104537)  (10000000,0.464136967475858)  (15000000,0.292554255468199)  (20000000,2.75335310107039E-14)  (25000000,2.75335310107039E-14)  (26500000,2.75335310107039E-14)  (28000000,2.75335310107039E-14)  (30000000,2.75335310107039E-14)  (70000000,2.75335310107039E-14)  (100000000,2.75335310107039E-14)  };
\addplot coordinates {  (100,0)  (500,0.00187818892159206)  (1000,-0.000417074688244912)  (5000,-0.000209297892355589)  (10000,0.00118597281292042)  (50000,0.00622459987747334)  (100000,0.00830567012845707)  (500000,0.00739853873937257)  (1000000,0.00919972275159225)  (5000000,0.0151476271295558)  (10000000,0.0184206511224776)  (15000000,0.292554255468199)  (20000000,2.75335310107039E-14)  (25000000,2.75335310107039E-14)  (26500000,2.75335310107039E-14)  (28000000,2.75335310107039E-14)  (30000000,2.75335310107039E-14)  (70000000,2.75335310107039E-14)  (100000000,2.75335310107039E-14)  };

   \end{semilogxaxis} 
  \end{tikzpicture}

\label{fig:RockYouResults95}}
\subfloat[RockYou: $\hat{v} \neq v$]{
    \begin{tikzpicture}[scale=1.0]  
   \begin{semilogxaxis}[
    title style={align=center},
    title={ {\small $\alpha=0.95, \hat{v} = C_{max}\times 10^4$ }},
       xlabel={Adversary Value: $\frac{v}{C_{max}}$},
    ylabel shift = -3pt,
    grid=major,
    small,
    cycle list = {{red, mark=none}, {blue, mark=none}, {green, mark=none},{blue, dashed, mark=none}, {green, dashed, mark=none},  {purple, dotted, mark=none}, {red, dashed, mark=none},{brown, dashed, mark=none} },
    legend style = {font=\tiny, at={(.05,.95)}, anchor=north west},
    legend entries = { $\PAdvCASH{v}{\hat{v}}{C_{max}}$, $\PAdvUnif{v}{C_{max}}$, $\PAdvDet{v}{C_{max}}$, Improvement: \textcolor{blue}{blue}$-$\textcolor{red}{red}, Improvement: \textcolor{green}{green}$-$\textcolor{red}{red}, $v = \hat{v}$  }
   ] 

    \addlegendimage{no markers, red, mark=square*}
    \addlegendimage{no markers, blue, mark=square*}
    \addlegendimage{no markers, green, mark=square*}
    \addlegendimage{no markers, dashed, blue}
    \addlegendimage{no markers, dashed, green}
    \addlegendimage{no markers, dotted, purple}
\addplot coordinates {  (100,0)  (500,0.00703895183690825)  (1000,0.0145237093081992)  (5000,0.0278110594326778)  (10000,0.0454018848661649)  (50000,0.118439836836895)  (100000,0.155703757358036)  (500000,0.257165070416829)  (1000000,0.313403406714585)  (5000000,0.450583879056243)  (10000000,0.545506415413783)  (15000000,0.999999999999972)  (20000000,0.999999999999972)  (25000000,0.999999999999972)  (26500000,0.999999999999972)  (27000000,0.999999999999972)  (27500000,0.999999999999972)  (28000000,0.999999999999972)  (29000000,0.999999999999972)  (30000000,0.999999999999972)  (70000000,0.999999999999972)  (100000000,0.999999999999972)  };
\addplot coordinates {  (100,0)  (500,0.00891714075850031)  (1000,0.0136977788934083)  (5000,0.0265629142590948)  (10000,0.0465878576790853)  (50000,0.124664436714369)  (100000,0.164009427486493)  (500000,0.264563609156202)  (1000000,0.322603129466177)  (5000000,0.465731506185799)  (10000000,0.55428368364662)  (15000000,1)  (20000000,1)  (25000000,1)  (26500000,1)  (27000000,1)  (27500000,1)  (28000000,1)  (29000000,1)  (30000000,1)  (70000000,1)  (100000000,1)  };
 \addplot coordinates {  (100,0)  (500,0.0136977788934083)  (1000,0.0180747779954648)  (5000,0.0441965417827129)  (10000,0.0730608733055595)  (50000,0.160089926850547)  (100000,0.199287847017616)  (500000,0.317302545367371)  (1000000,0.376634078642379)  (5000000,0.518411276766697)  (10000000,1)  (15000000,1)  (20000000,1)  (25000000,1)  (26500000,1)  (27000000,1)  (27500000,1)  (28000000,1)  (29000000,1)  (30000000,1)  (70000000,1)  (100000000,1)  };

\addplot coordinates {  (100,0)  (500,0.00187818892159206)  (1000,-0.000825930414790916)  (5000,-0.00124814517358296)  (10000,0.00118597281292042)  (50000,0.00622459987747334)  (100000,0.00830567012845707)  (500000,0.00739853873937257)  (1000000,0.00919972275159225)  (5000000,0.0151476271295558)  (10000000,0.00877726823283698)  (15000000,2.75335310107039E-14)  (20000000,2.75335310107039E-14)  (25000000,2.75335310107039E-14)  (26500000,2.75335310107039E-14)  (27000000,2.75335310107039E-14)  (27500000,2.75335310107039E-14)  (28000000,2.75335310107039E-14)  (29000000,2.75335310107039E-14)  (30000000,2.75335310107039E-14)  (70000000,2.75335310107039E-14)  (100000000,2.75335310107039E-14)  };
\addplot coordinates {  (100,0)  (500,0.0066588270565)  (1000,0.00355106868726559)  (5000,0.0163854823500351)  (10000,0.0276589884393946)  (50000,0.0416500900136517)  (100000,0.0435840896595805)  (500000,0.0601374749505419)  (1000000,0.0632306719277944)  (5000000,0.0678273977104537)  (10000000,0.454493584586217)  (15000000,2.75335310107039E-14)  (20000000,2.75335310107039E-14)  (25000000,2.75335310107039E-14)  (26500000,2.75335310107039E-14)  (27000000,2.75335310107039E-14)  (27500000,2.75335310107039E-14)  (28000000,2.75335310107039E-14)  (29000000,2.75335310107039E-14)  (30000000,2.75335310107039E-14)  (70000000,2.75335310107039E-14)  (100000000,2.75335310107039E-14)  };

\addplot coordinates { (10000, 0)  (10000, 0.5)  (10000, 1) };

   \end{semilogxaxis} 
  \end{tikzpicture}
  
\label{fig:RockYouRobust2Results95}
}
\subfloat[RockYou: $\hat{v} \neq v$.]{
    \begin{tikzpicture}[scale=1.0]  
   \begin{semilogxaxis}[
    title style={align=center},
    title={ {\small $\alpha=0.95, \hat{v} = 7\times C_{max} \times 10^7$ }},
       xlabel={Adversary Value: $\frac{v}{C_{max}}$},
    ylabel shift = -3pt,
    grid=major,
    small,
    cycle list = {{red, mark=none}, {blue, mark=none}, {green, mark=none},{blue, dashed, mark=none}, {green, dashed, mark=none},  {purple, dotted, mark=none}, {red, dashed, mark=none},{brown, dashed, mark=none} },
    legend style = {font=\tiny, at={(.05,.95)}, anchor=north west},
    legend entries = { $\PAdvCASH{v}{\hat{v}}{C_{max}}$, $\PAdvUnif{v}{C_{max}}$, $\PAdvDet{v}{C_{max}}$, Improvement: \textcolor{blue}{blue}$-$\textcolor{red}{red}, Improvement: \textcolor{green}{green}$-$\textcolor{red}{red}, $v = \hat{v}$  }
   ] 

    \addlegendimage{no markers, red, mark=square*}
    \addlegendimage{no markers, blue, mark=square*}
    \addlegendimage{no markers, green, mark=square*}
    \addlegendimage{no markers, dashed, blue}
    \addlegendimage{no markers, dashed, green}
    \addlegendimage{no markers, dotted, purple}
\addplot coordinates {  (100,0)  (500,0.00703895183690825)  (1000,0.0145237093081992)  (5000,0.0278110594326778)  (10000,0.0454018848661649)  (50000,0.118439836836895)  (100000,0.155703757358036)  (500000,0.257165070416829)  (1000000,0.313403406714585)  (5000000,0.450583879056243)  (10000000,0.545506415413783)  (15000000,0.999999999999972)  (20000000,0.999999999999972)  (25000000,0.999999999999972)  (26500000,0.999999999999972)  (27000000,0.999999999999972)  (27500000,0.999999999999972)  (28000000,0.999999999999972)  (29000000,0.999999999999972)  (30000000,0.999999999999972)  (70000000,0.999999999999972)  (100000000,0.999999999999972)  };

\addplot coordinates {  (100,0)  (500,0.00891714075850031)  (1000,0.0136977788934083)  (5000,0.0265629142590948)  (10000,0.0465878576790853)  (50000,0.124664436714369)  (100000,0.164009427486493)  (500000,0.264563609156202)  (1000000,0.322603129466177)  (5000000,0.465731506185799)  (10000000,0.55428368364662)  (15000000,1)  (20000000,1)  (25000000,1)  (26500000,1)  (27000000,1)  (27500000,1)  (28000000,1)  (29000000,1)  (30000000,1)  (70000000,1)  (100000000,1)  };
 \addplot coordinates {  (100,0)  (500,0.0136977788934083)  (1000,0.0180747779954648)  (5000,0.0441965417827129)  (10000,0.0730608733055595)  (50000,0.160089926850547)  (100000,0.199287847017616)  (500000,0.317302545367371)  (1000000,0.376634078642379)  (5000000,0.518411276766697)  (10000000,1)  (15000000,1)  (20000000,1)  (25000000,1)  (26500000,1)  (27000000,1)  (27500000,1)  (28000000,1)  (29000000,1)  (30000000,1)  (70000000,1)  (100000000,1)  };

\addplot coordinates {  (100,0)  (500,0.00187818892159206)  (1000,-0.000825930414790916)  (5000,-0.00124814517358296)  (10000,0.00118597281292042)  (50000,0.00622459987747334)  (100000,0.00830567012845707)  (500000,0.00739853873937257)  (1000000,0.00919972275159225)  (5000000,0.0151476271295558)  (10000000,0.00877726823283698)  (15000000,2.75335310107039E-14)  (20000000,2.75335310107039E-14)  (25000000,2.75335310107039E-14)  (26500000,2.75335310107039E-14)  (27000000,2.75335310107039E-14)  (27500000,2.75335310107039E-14)  (28000000,2.75335310107039E-14)  (29000000,2.75335310107039E-14)  (30000000,2.75335310107039E-14)  (70000000,2.75335310107039E-14)  (100000000,2.75335310107039E-14)  };
\addplot coordinates {  (100,0)  (500,0.0066588270565)  (1000,0.00355106868726559)  (5000,0.0163854823500351)  (10000,0.0276589884393946)  (50000,0.0416500900136517)  (100000,0.0435840896595805)  (500000,0.0601374749505419)  (1000000,0.0632306719277944)  (5000000,0.0678273977104537)  (10000000,0.454493584586217)  (15000000,2.75335310107039E-14)  (20000000,2.75335310107039E-14)  (25000000,2.75335310107039E-14)  (26500000,2.75335310107039E-14)  (27000000,2.75335310107039E-14)  (27500000,2.75335310107039E-14)  (28000000,2.75335310107039E-14)  (29000000,2.75335310107039E-14)  (30000000,2.75335310107039E-14)  (70000000,2.75335310107039E-14)  (100000000,2.75335310107039E-14)  };

\addplot coordinates { (70000000, 0)  (70000000, 0.5)  (70000000, 1) };
   \end{semilogxaxis} 
  \end{tikzpicture}
  
\label{fig:RockYouRobust1Results95}
}
\caption{$\alpha=0.95$}

\end{figure}

\subsection{Discussion} 
In our experiments we varied the password correctness rate $\alpha \in \{0.9,0.95,1\}$. Intuitively, we expect for CASH to have a greater advantage over traditional key-stretching techniques when $\alpha$ is larger, but when $\alpha \rightarrow 0$ we should not expect for CASH or uniform-CASH to outperform deterministic key-stretching techniques because there is no advantage in making authentication costs asymmetric. It is easier for users to remember passwords that they use frequently\cite{Pimsleur1967,BS14,blockiSpacedRepetition} so we would expect for $\alpha$ to be larger for services that are used frequently (e.g., e-mail). This suggests that larger values of $\alpha$ (e.g., $\alpha = 0.9$ or $\alpha=0.95$) would be appropriate for many services because the users who authenticate most frequently would be the least likely to enter incorrect passwords. While different authentication servers might experience different failed login rates $1-\alpha$, we remark that it is reasonable to assume that the authentication server knows the value of $\alpha$ because it can monitor login attempts.

\paragraph{Estimating $v$} While our results suggest that CASH continues to perform well even if our estimate $\hat{v}$ of the adversary's value $v$ for cracked passwords is wrong, we would still recommend that an authentication server perform a careful economic analysis to obtain the estimate $\hat{v}$ before running Algorithm \ref{alg:HeuristicAlg2} to compute the CASH distribution $\tilde{p}$. The organization should take into account empirical data on the cost $\mathbf{Cost}\left(\mathbf{H}\right)$ of computing the underlying hash function as well as the market value of a cracked password. If possible, we recommend that the organization consider data from black market sales of passwords for similar types of  accounts (e.g., an adversary would likely value a cracked Bank of America password more than a cracked Twitter password). Symantec reports that cracked passwords are sold on the black market for $\$4$--$\$30$~\cite{passwordBlackMarket}. Thus, $\$30/\CostH$ might be a reasonable upper bound on the adversary's value for a cracked password (measured in \# of computations of $\Hash^k$). We would also strongly advocate for the use of memory hard functions instead of hash iteration to increase $\CostH$ effectively (see discussion in Section \ref{subsec:Traditional}).

 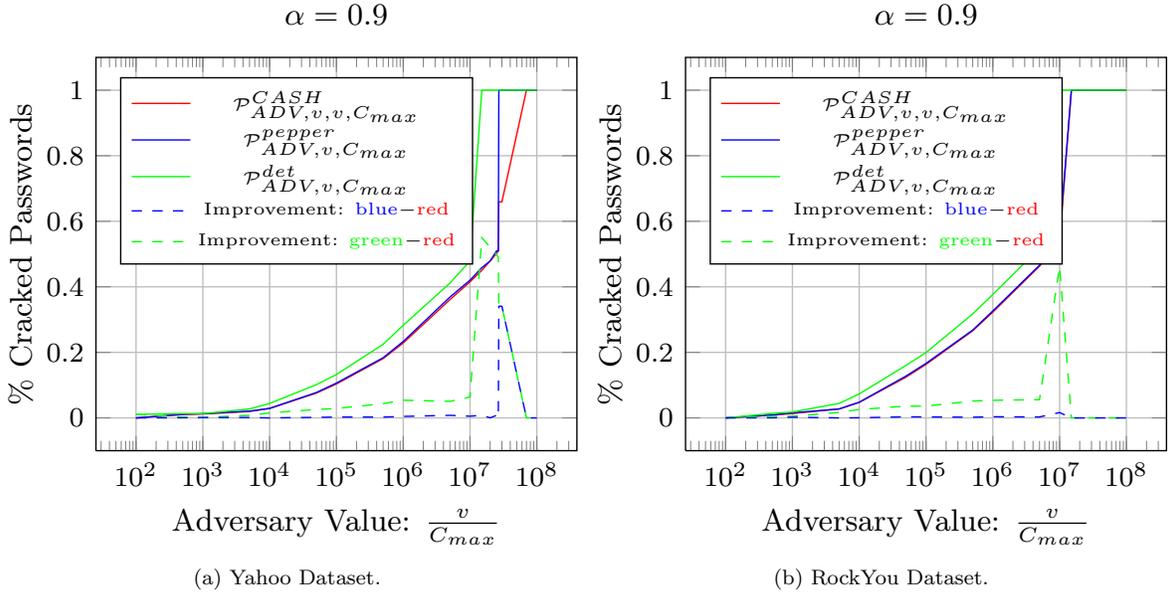
\begin{figure}[!t]
  \subfloat[Yahoo Dataset.]{
      \begin{tikzpicture}[scale=1.3]  
   \begin{semilogxaxis}[
    title style={align=center},
    title={ {\small $\alpha=0.9$ }},
    xlabel={Adversary Value: $\frac{v}{C_{max}}$},
    ylabel={$\%$ Cracked Passwords},
    ylabel shift = -3pt,
    grid=major,
    small,
    cycle list = {{red, mark=none}, {blue, mark=none}, {green, mark=none},{green, dashed, mark=none}, {blue, dashed, mark=none},  {blue, dashed, mark=none}, {red, dashed, mark=none},{brown, dashed, mark=none} },
    legend style = {font=\tiny, at={(.05,.95)}, anchor=north west},
    legend entries = { $\PAdvCASH{v}{v}{C_{max}}$, $\PAdvUnif{v}{C_{max}}$, $\PAdvDet{v}{C_{max}}$, Improvement: \textcolor{blue}{blue}$-$\textcolor{red}{red}, Improvement: \textcolor{green}{green}$-$\textcolor{red}{red}}
   ] 

    \addlegendimage{no markers, red, mark=square*}
    \addlegendimage{no markers, blue, mark=square*}
    \addlegendimage{no markers, green, mark=square*}
    \addlegendimage{no markers, dashed, blue}
    \addlegendimage{no markers, dashed, green}
\addplot coordinates {  (100,0)  (500,0.0103650003300184)  (1000,0.0118193070978086)  (5000,0.0196537617974582)  (10000,0.0289970884688705)  (50000,0.0757213976299667)  (100000,0.103595972530575)  (500000,0.181006843045768)  (1000000,0.228656981329561)  (5000000,0.360821094325625)  (10000000,0.415177925518485)  (15000000,0.450283451843314)  (20000000,0.47863343472216)  (25000000,0.501829516819531)  (26500000,0.509750280863969)  (27000000,0.659338193697864)  (27500000,0.659338193697864)  (28000000,0.659338193697864)  (29000000,0.659338193697864)  (30000000,0.659338193697864)  (70000000,0.999999999999979)  (100000000,0.999999999999979)  };
\addplot coordinates {  (100,0)  (500,0.0108687224894377)  (1000,0.0130192581998815)  (5000,0.0207614032035197)  (10000,0.0289970884688704)  (50000,0.0773727496772537)  (100000,0.106225338769438)  (500000,0.183571148129509)  (1000000,0.233240940214472)  (5000000,0.368658241037976)  (10000000,0.420317345392629)  (15000000,0.45834913690049)  (20000000,0.478853387778074)  (25000000,0.50975028086399)  (26500000,0.50975028086399)  (27000000,1)  (27500000,1)  (28000000,1)  (29000000,1)  (30000000,1)  (70000000,1)  (100000000,1)  };
 \addplot coordinates {  (100,0.0108687224894377)  (500,0.0130192581998815)  (1000,0.0130192581998815)  (5000,0.0273838584095427)  (10000,0.0441510096695537)  (50000,0.101412040578669)  (100000,0.132372669808665)  (500000,0.224601121331902)  (1000000,0.282624504055384)  (5000000,0.411368138539665)  (10000000,0.478853387778074)  (15000000,1)  (20000000,1)  (25000000,1)  (26500000,1)  (27000000,1)  (27500000,1)  (28000000,1)  (29000000,1)  (30000000,1)  (70000000,1)  (100000000,1)  };
\addplot coordinates {  (100,0.0108687224894377)  (500,0.00265425786986305)  (1000,0.00119995110207291)  (5000,0.00773009661208442)  (10000,0.0151539212006832)  (50000,0.025690642948702)  (100000,0.0287766972780896)  (500000,0.0435942782861336)  (1000000,0.053967522725823)  (5000000,0.0505470442140399)  (10000000,0.0636754622595894)  (15000000,0.549716548156686)  (20000000,0.52136656527784)  (25000000,0.498170483180469)  (26500000,0.490249719136031)  (27000000,0.340661806302136)  (27500000,0.340661806302136)  (28000000,0.340661806302136)  (29000000,0.340661806302136)  (30000000,0.340661806302136)  (70000000,2.14273043752655E-14)  (100000000,2.14273043752655E-14)  };
\addplot coordinates {  (100,0)  (500,0.000503722159419252)  (1000,0.00119995110207291)  (5000,0.00110764140606149)  (10000,-1.49186218934005E-16)  (50000,0.00165135204728702)  (100000,0.00262936623886303)  (500000,0.00256430508374061)  (1000000,0.00458395888491178)  (5000000,0.00783714671235103)  (10000000,0.00513941987414457)  (15000000,0.0080656850571757)  (20000000,0.000219953055913658)  (25000000,0.00792076404445874)  (26500000,2.1094237467878E-14)  (27000000,0.340661806302136)  (27500000,0.340661806302136)  (28000000,0.340661806302136)  (29000000,0.340661806302136)  (30000000,0.340661806302136)  (70000000,2.14273043752655E-14)  (100000000,2.14273043752655E-14)  };

   \end{semilogxaxis} 
  \end{tikzpicture}
  
\label{fig:YahooResults90}}
\subfloat[RockYou Dataset.]{
    \begin{tikzpicture}[scale=1.3]  
   \begin{semilogxaxis}[
    title style={align=center},
    title={ {\small $\alpha=0.9$ }},
    xlabel={Adversary Value: $\frac{v}{C_{max}}$},
    ylabel={$\%$ Cracked Passwords},
    ylabel shift = -3pt,
    grid=major,
    small,
    cycle list = {{red, mark=none}, {blue, mark=none}, {green, mark=none},{green, dashed, mark=none}, {blue, dashed, mark=none},  {blue, dashed, mark=none}, {red, dashed, mark=none},{brown, dashed, mark=none} },
    legend style = {font=\tiny, at={(.05,.95)}, anchor=north west},
    legend entries = { $\PAdvCASH{v}{v}{C_{max}}$, $\PAdvUnif{v}{C_{max}}$, $\PAdvDet{v}{C_{max}}$, Improvement: \textcolor{blue}{blue}$-$\textcolor{red}{red}, Improvement: \textcolor{green}{green}$-$\textcolor{red}{red}}
   ] 

    \addlegendimage{no markers, red, mark=square*}
    \addlegendimage{no markers, blue, mark=square*}
    \addlegendimage{no markers, green, mark=square*}
    \addlegendimage{no markers, dashed, blue}
    \addlegendimage{no markers, dashed, green}
\addplot coordinates {  (100,0)  (500,0.00859451942407725)  (1000,0.0133447400377569)  (5000,0.0272760303315716)  (10000,0.0474032707627518)  (50000,0.124170424428517)  (100000,0.163294928646465)  (500000,0.265801301128427)  (1000000,0.322799966527542)  (5000000,0.462803564928199)  (10000000,0.538000264865568)  (15000000,0.999999999999978)  (20000000,0.999999999999978)  (25000000,0.999999999999978)  (26500000,0.999999999999978)  (27000000,0.999999999999978)  (27500000,0.999999999999978)  (28000000,0.999999999999978)  (29000000,0.999999999999978)  (30000000,0.999999999999978)  (70000000,0.999999999999978)  (100000000,0.999999999999978)  };
\addplot coordinates {  (100,0)  (500,0.00891714075850031)  (1000,0.0155215770827253)  (5000,0.0272760303315717)  (10000,0.0478123623225905)  (50000,0.127514539286531)  (100000,0.166347957457673)  (500000,0.267932308139264)  (1000000,0.326354365380677)  (5000000,0.465731506185799)  (10000000,0.55428368364662)  (15000000,1)  (20000000,1)  (25000000,1)  (26500000,1)  (27000000,1)  (27500000,1)  (28000000,1)  (29000000,1)  (30000000,1)  (70000000,1)  (100000000,1)  };
 \addplot coordinates {  (100,0)  (500,0.0136977788934083)  (1000,0.0180747779954648)  (5000,0.0441965417827129)  (10000,0.0730608733055595)  (50000,0.160089926850547)  (100000,0.199287847017616)  (500000,0.317302545367371)  (1000000,0.376634078642379)  (5000000,0.518411276766697)  (10000000,1)  (15000000,1)  (20000000,1)  (25000000,1)  (26500000,1)  (27000000,1)  (27500000,1)  (28000000,1)  (29000000,1)  (30000000,1)  (70000000,1)  (100000000,1)  };
\addplot coordinates {  (100,0)  (500,0.00510325946933101)  (1000,0.00473003795770787)  (5000,0.0169205114511413)  (10000,0.0256576025428077)  (50000,0.0359195024220297)  (100000,0.0359929183711517)  (500000,0.0515012442389441)  (1000000,0.0538341121148372)  (5000000,0.0556077118384979)  (10000000,0.461999735134432)  (15000000,2.22044604925031E-14)  (20000000,2.22044604925031E-14)  (25000000,2.22044604925031E-14)  (26500000,2.22044604925031E-14)  (27000000,2.22044604925031E-14)  (27500000,2.22044604925031E-14)  (28000000,2.22044604925031E-14)  (29000000,2.22044604925031E-14)  (30000000,2.22044604925031E-14)  (70000000,2.22044604925031E-14)  (100000000,2.22044604925031E-14)  };
\addplot coordinates {  (100,0)  (500,0.000322621334423063)  (1000,0.00217683704496837)  (5000,8.32667268468867E-17)  (10000,0.000409091559838717)  (50000,0.00334411485801335)  (100000,0.00305302881120795)  (500000,0.0021310070108374)  (1000000,0.00355439885313574)  (5000000,0.00292794125760004)  (10000000,0.0162834187810521)  (15000000,2.22044604925031E-14)  (20000000,2.22044604925031E-14)  (25000000,2.22044604925031E-14)  (26500000,2.22044604925031E-14)  (27000000,2.22044604925031E-14)  (27500000,2.22044604925031E-14)  (28000000,2.22044604925031E-14)  (29000000,2.22044604925031E-14)  (30000000,2.22044604925031E-14)  (70000000,2.22044604925031E-14)  (100000000,2.22044604925031E-14)  };

   \end{semilogxaxis} 
  \end{tikzpicture}
  
\label{fig:RockYouResults90}}
\centering
\caption{$\alpha = 0.9$.}
\end{figure}

\paragraph{Obtaining an Empirical Password Distribution} We remark that the specific CASH distributions we computed for the RockYou and Yahoo! datasets might not be optimal in other application settings because the underlying password distribution may vary across different contexts. For example, users might be more motivated to pick strong passwords for higher value accounts (e.g., bank accounts). Similarly, some organizations choose to restrict the passwords that a user may select (e.g., requiring upper and lower case letters). While these restrictions do not always result in stronger passwords~\cite{usability:compositionPolicies}, they can alter the underlying password distribution~\cite{blockiPasswordComposition}. While the underlying distribution may vary from context to context, we note that an authentication server could always follow the framework of Bonneau~\cite{bonneau2012science} and Blocki et al.~\cite{blocki2016differentially} to securely approximate the password distribution $p_1,\ldots,p_n$ of its own users. 

If an organization remains highly uncertain about value $v$ of a cracked password or about the empirical password distribution $p_1,\ldots, p_n$ then it may be prudent to adopt the uniform-CASH mechanism (e.g., \cite{manber1996simple}), which {\em always } performs at least as well as the traditional key-stretching approach.

\subsubsection{Experimental Limitations} \label{subsubsec:Limitations}
We remark that values of $\PAdvCASH{v}{\hat{v}}{C}$ that we compute in our experiments may be less realistic for larger values of $\frac{v}{\ServerCost{\alpha}}$ (e.g., $10^8$). The reason is that $p_i$, our empirical estimate of the probability of password $pwd_i$, will be too high for many of our unique passwords in the dataset. For example, consider a dedicated user who memorizes a truly random $20$ character string of upper and lower case letters. The true probability that any individual password guess matches the user's password would be at most $1/52^{20} \approx 1/\left(2.09 \times 10^{34} \right) $. However, if that password occurred in the RockYou dataset then our empirical estimate of this probability would be at least $1/\left(3.26\times 10^7\right)$. Developing improved techniques for estimating the true likelihood of unique password in a password frequency dataset is an important research direction.

\section{Related Work} \label{sec:Related}
\noindent {\bf Breaches.} Recent breaches~\cite{breach:Adobe,breach:IEEE,breach:Zappos,breach:ebay,breach:linkedin,breach:sony,breach:rockyou} highlight the importance of proper password storage. In one of these instances~\cite{breach:rockyou} passwords were stored on the authentication server in cleartext and in other instances the passwords were not salted~\cite{breach:linkedin}. Salting is a simple, yet effective, way to defend against rainbow table attacks~\cite{salt}, which can be used to dramatically reduce the cost of an offline attack against unsalted passwords~\cite{rainbowTable}. Bonneau and Preibusch~\cite{bonneau2010password} found that implementation errors like these are unfortunately commonplace. 

\noindent {\bf Key Stretching.} The practice of key stretching was proposed as early as 1979 by Morris and Thompson~\cite{morris1979password}. The goal is to make the hash function more expensive to evaluate so that an offline attack is more expensive for the adversary. PBKDF2~\cite{kaliski2000pkcs}, BCRYPT~\cite{bcrypt} use hash iteration to accomplish this goal. The recent Ashley Madison breach highlights the benefits of key-stretching in practice. Through an implementation mistake half of the Ashley Madison passwords were protected with the MD5 hash function instead of the much stronger BCRYPT hash function allowing offline password crackers to quickly recover these passwords\footnote{See,  \url{http://arstechnica.com/security/2015/09/once-seen-as-bulletproof-11-million-ashley-madison-passwords-already-cracked/} (retrieved 5/4/2016)}.

More modern password hash functions like SCRYPT~\cite{percival2012scrypt} use memory hard functions for key-stretching. Recently, the Password Hashing Competition\cite{PHC} was developed to encourage the development of alternative password hashing schemes (e.g., \cite{almeida2014lyra,forler2013catena}). Argon2~\cite{Argon2}, the winner, has a parameters which control memory usage and parallelism. Deterministic key-stretching methods result in proportionally increased costs for the legitimate server as well as the adversary. Manber~\cite{manber1996simple} proposed the use of hidden salt values (e.g., `pepper') to make it more expensive to reject incorrect passwords. CASH may be viewed as a generalization of this idea. Boyen~\cite{boyen2007halting} proposed using halting puzzles to introduce an extreme asymmetry --- the password verification algorithm never halts when we try an incorrect password. However, in practice an authentication server will need to upper bound the maximum running time for authentication because even legitimate users may occasionally enter the wrong password. 

\noindent {\bf Other Defenses Against Offline Attacks.} If an organization has multiple servers for authentication then it is possible to distribute storage of the passwords across multiple servers to keep them safe from an adversary who only breaches one server (e.g., see \cite{brainard2003new} or ~\cite{camenisch2012practical}). Juels and Rivest~\cite{rivestHoneywords} proposed storing the hashes of fake passwords (honeywords) and using a second auxiliary server to detect an offline attack (authentication attempt with honeywords). Another line of research has sought to include the solution(s) to hard artificial intelligence problems in the password hash so that an offline attacker needs human assistance~\cite{canetti2006mitigating,poshExperiment,blockiGOTCHA}. By contrast, CASH does not require an organization to purchase and maintain multiple (distributed) authentication servers and it could be adopted without altering the user's authentication experience (e.g., by requiring the user to solve CAPTCHAs). 

\noindent {\bf Measuring Password Strength.} Guessing-entropy~\cite{shannon1959mathematical,massey1994guessing}, $\sum_{i=1}^n i\times p_i$, measures the average number of guesses needed to crack a single password. We use a similar formula to compute how much work a threshold-$B$ adversary would do in expectation. Guessing-entropy and Shannon-entropy are known to be poor metrics for measuring password strength\footnote{Guessing-entropy could be high even if half of our users choose the same password ($p_1 = 0.5$) as long as the other half of our users choose a password uniformly at random from $\PasswordSpace$ $\left(p_2 = \ldots = p_n = \frac{2}{n-1}\right)$. } While minimum entropy, $H_{\infty}=-\log p_1$, can be used to estimate the fraction of passwords that could be  cracked in an online attack~\cite{blockiPasswordComposition}, it can provide an overly pessimistic security measurement in general. 

 Boztas~\cite{boztas1999entropies} proposed a metric called $\beta$-guesswork, which measured the success rate for an adversary with $\beta$ guesses per account $\sum_{i=1}^{\beta} p_i$. We use a similar formula for computing the success rate of a threshold-$B$ adversary  against our CASH mechanism --- the key difference is that the adversary must guess the random value $t_u$ as well as the user's password $pwd_u$. Pliam's proposed a similar metric called $\alpha$-guesswork~\cite{pliam2000incomparability}, which measures the number of password guesses the adversary would need  (per user) to achieve success rate $\alpha$. 

\noindent {\bf Encouraging Users to Memorize Stronger Passwords.} A separate line of research has focused on helping users memorize stronger passwords using various mnemonic techniques and/or rehearsal techniques (e.g.,~\cite{blockiNaturallyRehearsingPasswords,BS14,usabilitystudy:xkcd,Hertzum:2006:minimalFeedbackHintsforPasswords}). 

Password managers seek to minimize user burden by using a single password to generate multiple passwords~\cite{ross2005stronger}. These password managers often use client-side key stretching to derive each password. While CASH is a useful tool for server-side key stretching, our current version of CASH is not appropriate for client-side key stretching because the authentication procedure is not deterministic. In subsequent work, Blocki and Sridhar~\cite{blocki2016client} developed Client-CASH an extension of CASH suitable for client-side key stretching.

\noindent{\bf Password Alternatives.} Another line of research has focused on developing alternatives to text passwords like graphical passwords\cite{Jermyn:1999:DAG:1251421.1251422,brostoff2000passfaces,biddle2012graphical}. Herley and van Oorschoot argued that text passwords will remain the dominant means of authentication despite attempts to replace them~ \cite{passwordPersistence}. We note that CASH could be used to protect graphical passwords as well as text passwords.

\section{Conclusions} \label{sec:Conclusions}
We presented a novel Stackelberg game model which captures the essential elements of the interaction between an authentication server (leader) and an offline password cracker (follower). Our Stackelberg model can provide guidance for the authentication server by providing an estimate of how significantly key-stretching reduces the number of passwords that would be cracked by a rational offline adversary in the event of a server breach. We also introduced, CASH, a randomized secure hashing algorithm that significantly outperforms traditional key-stretching defenses in our Stackelberg game. While the problem of computing an exact Stackelberg equlibria is non-convex, we were able to find an efficient heuristic algorithm to compute good strategies for the authentication server. Our heuristic algorithm is based on a highly related problem that we are able to show is tractable. Finally, we analyzed the performance of our CASH mechanism using empirical password data from two large scale password frequency datasets: Yahoo! and RockYou. Our empirical analysis demonstrates that the CASH mechanism can significantly (e.g., $50\%$) reduce the fraction of passwords that would be cracked in an offline attack by a rational adversary. Thus, our CASH mechanism can be used to mitigate the threat of offline attacks without increasing computation costs for a legitimate authentication server.



\section*{Acknowledgments}
This work was completed in part while the first author was visiting the Simons Institute for the Theory of Computing, supported by the Simons Foundation and by the DIMACS/Simons
Collaboration in Cryptography through NSF grant $\#$CNS-1523467. The research was also supported by an AFOSR MURI on Science of Cybersecurity as well as grants from the NSF TRUST Center. The views expressed in this paper are those of the authors and do not necessarily reflect the views of the Simons Institute or the National Science Foundation.

%
%

\bibliographystyle{IEEEtranS}
\bibliography{password,bounded-parallel-mhf}

\begin{thebibliography}{10}
\providecommand{\url}[1]{#1}
\csname url@samestyle\endcsname
\providecommand{\newblock}{\relax}
\providecommand{\bibinfo}[2]{#2}
\providecommand{\BIBentrySTDinterwordspacing}{\spaceskip=0pt\relax}
\providecommand{\BIBentryALTinterwordstretchfactor}{4}
\providecommand{\BIBentryALTinterwordspacing}{\spaceskip=\fontdimen2\font plus
\BIBentryALTinterwordstretchfactor\fontdimen3\font minus
  \fontdimen4\font\relax}
\providecommand{\BIBforeignlanguage}[2]{{%
\expandafter\ifx\csname l@#1\endcsname\relax
\typeout{** WARNING: IEEEtranS.bst: No hyphenation pattern has been}%
\typeout{** loaded for the language `#1'. Using the pattern for}%
\typeout{** the default language instead.}%
\else
\language=\csname l@#1\endcsname
\fi
#2}}
\providecommand{\BIBdecl}{\relax}
\BIBdecl

\bibitem{PHC}
``Password hashing competition,'' \url{https://password-hashing.net/}.

\bibitem{breach:rockyou}
``Rockyou hack: From bad to worse,''
  http://techcrunch.com/2009/12/14/rockyou-hack-security-myspace-facebook-passwords/,
  December 2009, retrieved 9/27/2012.

\bibitem{breach:sony}
``Update on playstation network/qriocity services,''
  \url{http://blog.us.playstation.com/2011/04/22/update-on-playstation-network-qriocity-services/},
  April 2011, retrieved 5/22/2012.

\bibitem{breach:IEEE}
``Data breach at ieee.org: 100k plaintext passwords,''
  \url{http://ieeelog.com/}, September 2012, retrieved 9/27/2012.

\bibitem{breach:linkedin}
``An update on linkedin member passwords compromised,''
  \url{http://blog.linkedin.com/2012/06/06/linkedin-member-passwords-compromised/},
  June 2012, retrieved 9/27/2012.

\bibitem{breach:Zappos}
``Zappos customer accounts breached,''
  \url{http://www.usatoday.com/tech/news/story/2012-01-16/mark-smith-zappos-breach-tips/52593484/1
  }, January 2012, retrieved 5/22/2012.

\bibitem{breach:Adobe}
``Important customer security announcement,''
  \url{http://blogs.adobe.com/conversations/2013/10/important-customer-security-announcement.html},
  October 2013, retrieved 2/10/2014.

\bibitem{breach:ebay}
``ebay suffers massive security breach, all users must change their
  passwords,'' \url
  http://www.forbes.com/sites/gordonkelly/2014/05/21/ebay-suffers-massive-security-breach-all-users-must-their-change-passwords/,
  May 2014.

\bibitem{salt}
S.~Alexander, ``Password protection for modern operating systems,''
  \emph{;login}, June 2004.

\bibitem{almeida2014lyra}
L.~C. Almeida, E.~R. Andrade, P.~S. Barreto, and M.~A. Simplicio~Jr, ``Lyra:
  Password-based key derivation with tunable memory and processing costs,''
  \emph{Journal of Cryptographic Engineering}, vol.~4, no.~2, pp. 75--89, 2014.

\bibitem{biddle2012graphical}
R.~Biddle, S.~Chiasson, and P.~Van~Oorschot, ``Graphical passwords: Learning
  from the first twelve years,'' \emph{ACM Computing Surveys (CSUR)}, vol.~44,
  no.~4, p.~19, 2012.

\bibitem{Argon2}
A.~Biryukov, D.~Dinu, and D.~Khovratovich, ``Fast and tradeoff-resilient
  memory-hard functions for cryptocurrencies and password hashing,'' Cryptology
  ePrint Archive, Report 2015/430, 2015, \url{http://eprint.iacr.org/}.

\bibitem{blockiGOTCHA}
J.~Blocki, M.~Blum, and A.~Datta, ``Gotcha password hackers!'' in
  \emph{Proceedings of the 2013 ACM workshop on Artificial intelligence and
  security}.\hskip 1em plus 0.5em minus 0.4em\relax ACM, 2013, pp. 25--34.

\bibitem{blockiNaturallyRehearsingPasswords}
\BIBentryALTinterwordspacing
------, ``Naturally rehearsing passwords,'' in \emph{Advances in Cryptology -
  ASIACRYPT 2013}, ser. Lecture Notes in Computer Science, K.~Sako and
  P.~Sarkar, Eds., vol. 8270.\hskip 1em plus 0.5em minus 0.4em\relax Springer
  Berlin Heidelberg, 2013, pp. 361--380. [Online]. Available:
  \url{http://dx.doi.org/10.1007/978-3-642-42045-0_19}
\BIBentrySTDinterwordspacing

\bibitem{blocki2013audit}
J.~Blocki, N.~Christin, A.~Datta, A.~D. Procaccia, and A.~Sinha, ``Audit
  games,'' in \emph{Proceedings of the Twenty-Third international joint
  conference on Artificial Intelligence}.\hskip 1em plus 0.5em minus
  0.4em\relax AAAI Press, 2013, pp. 41--47.

\bibitem{blocki2016differentially}
J.~Blocki, A.~Datta, and J.~Bonneau, ``Differentially private password
  frequency lists,'' in \emph{23rd Annual Network and Distributed System
  Security Symposium, {NDSS} 2016}, 2016.

\bibitem{blockiSpacedRepetition}
\BIBentryALTinterwordspacing
J.~Blocki, S.~Komanduri, L.~F. Cranor, and A.~Datta, ``Spaced repetition and
  mnemonics enable recall of multiple strong passwords,'' \emph{CoRR}, vol.
  abs/1410.1490, 2014. [Online]. Available:
  \url{http://arxiv.org/abs/1410.1490}
\BIBentrySTDinterwordspacing

\bibitem{blockiPasswordComposition}
J.~Blocki, S.~Komanduri, A.~Procaccia, and O.~Sheffet, ``Optimizing password
  composition policies,'' in \emph{Proceedings of the fourteenth ACM conference
  on Electronic commerce}.\hskip 1em plus 0.5em minus 0.4em\relax ACM, 2013,
  pp. 105--122.

\bibitem{blocki2016client}
J.~Blocki and A.~Sridhar, ``Client-cash: Protecting master passwords against
  offline attacks,'' in \emph{AsiaCCS (to appear)}.\hskip 1em plus 0.5em minus
  0.4em\relax ACM, 2016.

\bibitem{bonneau2012science}
J.~Bonneau, ``The science of guessing: analyzing an anonymized corpus of 70
  million passwords,'' in \emph{Security and Privacy (SP), 2012 IEEE Symposium
  on}.\hskip 1em plus 0.5em minus 0.4em\relax IEEE, 2012, pp. 538--552.

\bibitem{bonneau2010password}
J.~Bonneau and S.~Preibusch, ``The password thicket: technical and market
  failures in human authentication on the web,'' in \emph{Proc. of WEIS}, vol.
  2010, 2010.

\bibitem{BS14}
J.~Bonneau and S.~Schechter, ``"toward reliable storage of 56-bit keys in human
  memory",'' in \emph{Proceedings of the 23rd USENIX Security Symposium},
  August 2014.

\bibitem{boyen2007halting}
X.~Boyen, ``Halting password puzzles,'' in \emph{Proc. Usenix Security}, 2007.

\bibitem{boztas1999entropies}
S.~Boztas, ``Entropies, guessing, and cryptography,'' \emph{Department of
  Mathematics, Royal Melbourne Institute of Technology, Tech. Rep}, vol.~6,
  1999.

\bibitem{brainard2003new}
J.~G. Brainard, A.~Juels, B.~Kaliski, and M.~Szydlo, ``A new two-server
  approach for authentication with short secrets.'' in \emph{USENIX Security},
  vol.~3, 2003, pp. 201--214.

\bibitem{brostoff2000passfaces}
S.~Brostoff and M.~Sasse, ``{Are Passfaces more usable than passwords: A field
  trial investigation},'' in \emph{People and Computers XIV-Usability or Else:
  Proceedings of HCI}, 2000, pp. 405--424.

\bibitem{camenisch2012practical}
J.~Camenisch, A.~Lysyanskaya, and G.~Neven, ``Practical yet universally
  composable two-server password-authenticated secret sharing,'' in
  \emph{Proceedings of the 2012 ACM conference on Computer and Communications
  Security}.\hskip 1em plus 0.5em minus 0.4em\relax ACM, 2012, pp. 525--536.

\bibitem{canetti2006mitigating}
R.~Canetti, S.~Halevi, and M.~Steiner, ``Mitigating dictionary attacks on
  password-protected local storage,'' in \emph{Advances in Cryptology-CRYPTO
  2006}.\hskip 1em plus 0.5em minus 0.4em\relax Springer, 2006, pp. 160--179.

\bibitem{conitzer2006computing}
V.~Conitzer and T.~Sandholm, ``Computing the optimal strategy to commit to,''
  in \emph{Proceedings of the 7th ACM Conference on Electronic Commerce}.\hskip
  1em plus 0.5em minus 0.4em\relax ACM, 2006, pp. 82--90.

\bibitem{poshExperiment}
W.~Daher and R.~Canetti, ``Posh: A generalized captcha with security
  applications,'' in \emph{Proceedings of the 1st ACM workshop on Workshop on
  AISec}.\hskip 1em plus 0.5em minus 0.4em\relax ACM, 2008, pp. 1--10.

\bibitem{JTR}
S.~Designer, ``{John the Ripper},'' \url{http://www.openwall.com/john/},
  1996-2010.

\bibitem{mostPopularPasswords2012}
\BIBentryALTinterwordspacing
K.~Doel, ``Scary logins: Worst passwords of 2012 — and how to fix them,''
  SplashData, 2012, retrieved 1/21/2013. [Online]. Available:
  \url{http://www.prweb.com/releases/2012/10/prweb10046001.htm}
\BIBentrySTDinterwordspacing

\bibitem{DGN03}
\BIBentryALTinterwordspacing
C.~Dwork, A.~Goldberg, and M.~Naor, ``On memory-bound functions for fighting
  spam,'' in \emph{Advances in Cryptology - CRYPTO 2003, 23rd Annual
  International Cryptology Conference, Santa Barbara, California, USA, August
  17-21, 2003, Proceedings}, ser. Lecture Notes in Computer Science, vol.
  2729.\hskip 1em plus 0.5em minus 0.4em\relax Springer, 2003, pp. 426--444.
  [Online]. Available:
  \url{http://www.iacr.org/cryptodb/archive/2003/CRYPTO/1266/1266.pdf}
\BIBentrySTDinterwordspacing

\bibitem{dwork2006calibrating}
C.~Dwork, F.~McSherry, K.~Nissim, and A.~Smith, ``Calibrating noise to
  sensitivity in private data analysis,'' in \emph{Theory of
  Cryptography}.\hskip 1em plus 0.5em minus 0.4em\relax Springer, 2006, pp.
  265--284.

\bibitem{forler2013catena}
C.~Forler, S.~Lucks, and J.~Wenzel, ``Catena: A memory-consuming password
  scrambler.'' \emph{IACR Cryptology ePrint Archive}, vol. 2013, p. 525, 2013.

\bibitem{passwordBlackMarket}
M.~Fossi, E.~Johnson, D.~Turner, T.~Mack, J.~Blackbird, D.~McKinney, M.~K. Low,
  T.~Adams, M.~P. Laucht, and J.~Gough, ``Symantec report on the underground
  economy,'' November 2008, retrieved 1/8/2013.

\bibitem{PasswordCrackingArticle}
D.~Goodin, ``Why passwords have never been weaker-and crackers have never been
  stronger,'' \url
  {http://arstechnica.com/security/2012/08/passwords-under-assault/}, August
  2012.

\bibitem{passwordPersistence}
C.~Herley and P.~van Oorschot, ``A research agenda acknowledging the
  persistence of passwords,'' \emph{IEEE Symposium Security and Privacy},
  vol.~10, no.~1, pp. 28--36, 2012.

\bibitem{Hertzum:2006:minimalFeedbackHintsforPasswords}
\BIBentryALTinterwordspacing
M.~Hertzum, ``Minimal-feedback hints for remembering passwords,''
  \emph{interactions}, vol.~13, pp. 38--40, May 2006. [Online]. Available:
  \url{http://doi.acm.org/10.1145/1125864.1125888}
\BIBentrySTDinterwordspacing

\bibitem{jain2010security}
M.~Jain, E.~Kardes, C.~Kiekintveld, F.~Ord{\'o}nez, and M.~Tambe, ``Security
  games with arbitrary schedules: A branch and price approach.'' in
  \emph{AAAI}, 2010.

\bibitem{Jermyn:1999:DAG:1251421.1251422}
\BIBentryALTinterwordspacing
I.~Jermyn, A.~Mayer, F.~Monrose, M.~K. Reiter, and A.~D. Rubin, ``The design
  and analysis of graphical passwords,'' in \emph{Proceedings of the 8th
  conference on USENIX Security Symposium - Volume 8}.\hskip 1em plus 0.5em
  minus 0.4em\relax Berkeley, CA, USA: USENIX Association, 1999, pp. 1--1.
  [Online]. Available:
  \url{http://portal.acm.org/citation.cfm?id=1251421.1251422}
\BIBentrySTDinterwordspacing

\bibitem{rivestHoneywords}
A.~Juels and R.~L. Rivest, ``Honeywords: Making password-cracking detectable,''
  in \emph{Proceedings of the 2012 ACM conference on Computer and
  communications security}.\hskip 1em plus 0.5em minus 0.4em\relax ACM, 2013.

\bibitem{kaliski2000pkcs}
B.~Kaliski, ``Pkcs\# 5: Password-based cryptography specification version
  2.0,'' 2000.

\bibitem{khachiyan1980polynomial}
L.~G. Khachiyan, ``Polynomial algorithms in linear programming,'' \emph{USSR
  Computational Mathematics and Mathematical Physics}, vol.~20, no.~1, pp.
  53--72, 1980.

\bibitem{usability:compositionPolicies}
S.~Komanduri, R.~Shay, P.~Kelley, M.~Mazurek, L.~Bauer, N.~Christin, L.~Cranor,
  and S.~Egelman, ``Of passwords and people: measuring the effect of
  password-composition policies,'' in \emph{Proceedings of the 2011 annual
  conference on Human factors in computing systems}.\hskip 1em plus 0.5em minus
  0.4em\relax ACM, 2011, pp. 2595--2604.

\bibitem{manber1996simple}
U.~Manber, ``A simple scheme to make passwords based on one-way functions much
  harder to crack,'' \emph{Computers \& Security}, vol.~15, no.~2, pp.
  171--176, 1996.

\bibitem{massey1994guessing}
J.~Massey, ``Guessing and entropy,'' in \emph{Information Theory, 1994.
  Proceedings., 1994 IEEE International Symposium on}.\hskip 1em plus 0.5em
  minus 0.4em\relax IEEE, 1994, p. 204.

\bibitem{morris1979password}
R.~Morris and K.~Thompson, ``Password security: A case history,''
  \emph{Communications of the ACM}, vol.~22, no.~11, pp. 594--597, 1979.

\bibitem{rainbowTable}
P.~Oechslin, ``Making a faster cryptanalytic time-memory trade-off,''
  \emph{Advances in Cryptology-CRYPTO 2003}, pp. 617--630, 2003.

\bibitem{Per09}
C.~Percival, ``Stronger key derivation via sequential memory-hard functions,''
  in \emph{BSDCan 2009}, 2009.

\bibitem{percival2012scrypt}
C.~Percival and S.~Josefsson, ``The scrypt password-based key derivation
  function,'' 2012.

\bibitem{Pimsleur1967}
\BIBentryALTinterwordspacing
P.~Pimsleur, ``\BIBforeignlanguage{English}{A memory schedule},''
  \emph{\BIBforeignlanguage{English}{The Modern Language Journal}}, vol.~51,
  no.~2, pp. pp. 73--75, 1967. [Online]. Available:
  \url{http://www.jstor.org/stable/321812}
\BIBentrySTDinterwordspacing

\bibitem{pliam2000incomparability}
J.~Pliam, ``On the incomparability of entropy and marginal guesswork in
  brute-force attacks,'' \emph{Progress in Cryptology�INDOCRYPT 2000}, pp.
  113--123, 2000.

\bibitem{bcrypt}
N.~Provos and D.~Mazieres, ``Bcrypt algorithm.''

\bibitem{ross2005stronger}
B.~Ross, C.~Jackson, N.~Miyake, D.~Boneh, and J.~C. Mitchell, ``Stronger
  password authentication using browser extensions.'' in \emph{Usenix
  security}.\hskip 1em plus 0.5em minus 0.4em\relax Baltimore, MD, USA, 2005,
  pp. 17--32.

\bibitem{seeley1989password}
D.~Seeley, ``Password cracking: A game of wits,'' \emph{Communications of the
  ACM}, vol.~32, no.~6, pp. 700--703, 1989.

\bibitem{shannon1959mathematical}
C.~Shannon and W.~Weaver, \emph{The mathematical theory of
  communication}.\hskip 1em plus 0.5em minus 0.4em\relax Citeseer, 1959.

\bibitem{usabilitystudy:xkcd}
R.~Shay, P.~Kelley, S.~Komanduri, M.~Mazurek, B.~Ur, T.~Vidas, L.~Bauer,
  N.~Christin, and L.~Cranor, ``Correct horse battery staple: Exploring the
  usability of system-assigned passphrases,'' in \emph{Proceedings of the
  Eighth Symposium on Usable Privacy and Security}.\hskip 1em plus 0.5em minus
  0.4em\relax ACM, 2012, p.~7.

\bibitem{yin2010stackelberg}
Z.~Yin, D.~Korzhyk, C.~Kiekintveld, V.~Conitzer, and M.~Tambe, ``Stackelberg
  vs. nash in security games: Interchangeability, equivalence, and
  uniqueness,'' in \emph{Proceedings of the 9th International Conference on
  Autonomous Agents and Multiagent Systems: volume 1-Volume 1}.\hskip 1em plus
  0.5em minus 0.4em\relax International Foundation for Autonomous Agents and
  Multiagent Systems, 2010, pp. 1139--1146.

\bibitem{zonenberg2009distributed}
A.~Zonenberg, ``Distributed hash cracker: A cross-platform gpu-accelerated
  password recovery system,'' \emph{Rensselaer Polytechnic Institute}, p.~27,
  2009.

\end{thebibliography}


\appendix

 \section*{Missing Proofs}

\begin{remindertheorem}{\ref{thm:Polytime}} \thmPolyTime
\end{remindertheorem}
\vspace{4mm}

\begin{proofof}{Theorem \ref{thm:Polytime}} (sketch) 
We first note that the convex feasible space from Optimization Goal \ref{goal:MinAdvSuccessAsLP} fits inside a ball of radius one. Thus, the Ellipsoid algorithm~\cite{khachiyan1980polynomial} will converge after making $poly(m)$ many queries to our separation oracle. By Theorem \ref{thm:SeparationOracle} the running time of the separation oracle is $O\left(mn \log mn\right)$. Thus, the total running time is polynomial in $m$ and $n$.
\end{proofof}

\subsection{Separation Oracle} \label{subsec:SeparationOracle}
The key idea behind Theorem \ref{thm:Polytime} is to develop a polynomial time separation oracle. A separation oracle is an algorithm that takes as input a convex set $K \subseteq \mathbb{R}^{m}$ and a point $p \in \mathbb{R}^m$. The separation oracle outputs ``Ok" if $p \in K$; otherwise it returns hyperplane separating $x$ from $K$. In our context, the separation oracle takes as input a proposed solution $\mathcal{P}_{Adv,B}',\tilde{p}_1', \ldots, \tilde{p}_m'$ and outputs ``Ok" if every constraint from Optimization Goal \ref{goal:MinAdvSuccessAsLP} is satisfied; otherwise the separation oracle finds a constraint that is not satisfied. If we can develop a polynomial time separation oracle for our linear program then we can use the ellipsoid algorithm to solve our linear program in polynomial time~\cite{khachiyan1980polynomial}. For our purposes, it is not necessary to understand how the ellipsoid algorithm works. Will we treat the ellipsoid algorithm as a blackbox that can solve a linear program in polynomial time given oracle access to a separation oracle. 

We now present a separation oracle for Goal \ref{goal:MinAdvSuccessAsLP}. Theorem \ref{thm:SeparationOracle} states that Algorithm \ref{alg:SeparationOracle} is a polynomial time separation oracle. We provide intuition for our separation oracle below. Theorem \ref{thm:Polytime} follows immediately because we can run the ellipsoid algorithm~\cite{khachiyan1980polynomial} with our separation oracle to solve Goal \ref{goal:MinAdvSuccessAsLP} in polynomial time. 

\newcommand{\thmSepOracle}{Algorithm \ref{alg:SeparationOracle}  is valid separation oracle for Goal \ref{goal:MinAdvSuccessAsLP} and runs in time $O\left(mn \log mn\right)$.}
\begin{theorem} \label{thm:SeparationOracle} \thmSepOracle \end{theorem}

\begin{algorithm}[H] 
 \caption{Separation Oracle. Output is an unsatisfied constraint $C$ or ``Ok" if every constraint is satisfied.}
\label{alg:SeparationOracle}
\begin{algorithmic}[1]
 \renewcommand{\algorithmicrequire}{\textbf{Input:}}
 \renewcommand{\algorithmicensure}{\textbf{Output:}}
\REQUIRE $p_1,\ldots, p_n$, $\tilde{p}_1',\ldots, \tilde{p}_m'$, $B$, $C_{max}$, $k$, $\alpha$, and and $\mathcal{P}_{Adv,B}'$.

\IF{$\sum_{i=1}^m \tilde{p}_i' \neq 1$ }
\RETURN $\sum_{i=1}^m \tilde{p}_i = 1$.
\ENDIF
\IF{$(1-\alpha)m\cdot k+ k\cdot\alpha \sum_{i=1}^m  i\cdot \tilde{p}_i'
 > C_{max}$}
\RETURN $(1-\alpha)m\cdot k+ k\cdot\alpha \sum_{i=1}^m  i\cdot \tilde{p}_i'
 \leq C_{max}$
\ENDIF
\FOR{i=1,\ldots,m}
\IF{$\tilde{p}_{i}' < 0$} 
\RETURN $\tilde{p}_i \geq 0$.
\ENDIF
\IF{ $i<m$ and $\tilde{p}_{i+1}' > \tilde{p}_{i}'$}
\RETURN $\tilde{p}_{i+1} \leq \tilde{p}_{i}$. 
\ENDIF
\ENDFOR
\IF{$\mathcal{P}_{Adv,B}' > 1$}
\RETURN $\mathcal{P}_{Adv,B}' \leq 1$
\ENDIF
\IF{$\mathcal{P}_{Adv,B}' < 0$} \RETURN $\mathcal{P}_{Adv,B}' \geq 0$
\ENDIF
\FOR{i = 1,\ldots, n}
\FOR{j = 1,\ldots, m}
\STATE $p_{i,j}' \gets p_i\tilde{p}_j'$. 
\ENDFOR
\ENDFOR

\STATE $TUPLES \gets \{ (i,j) ~\vline 1 \leq i \leq n \wedge 1 \leq j \leq m\}$.\\

\STATE {\bf Define} ordering $\succ$ over $TUPLES$:  $\left(i_1,j_1\right) \succ \left(i_2,j_2\right)$ if any of the following conditions hold (1) $p_{i_1,j_1} > p_{i_2,j_2}$, or (2) $p_{i_1,j_1} = p_{i_2,j_2}$ and $i_1 < i_2$ or (3)$p_{i_1,j_1} = p_{i_2,j_2}$ and $i_1 = i_2$ and $j_1 < j_2$. \\

\STATE $SORTED-TUPLES \gets \mathbf{SORT}\left(TUPLES,\succ \right)$. 
\COMMENT{Let  $T_k \doteq SORTED-TUPLES[k]$. } \\
\COMMENT{$T_k$ is the $k$\rq{}th biggest element according to  $\succ$}
\STATE $S \gets \{T_1,\ldots, T_B\}$.
\FOR{i=1,\ldots,n}
 \STATE $b_i' \gets \max\left\{j\in \mathbb{Z}~\vline j=0 \vee (i,j) \in S \right\}$ 
\ENDFOR
\IF{$\mathcal{P}_{Adv,B}' < \sum_{i=1}^n p_i \sum_{j=1}^{b_i'} \tilde{p}_j'$}
 
\RETURN $\mathcal{P}_{Adv,B} \geq \sum_{i=1}^n p_i \sum_{j=1}^{b_i'} \tilde{p}_j$
\ELSE
\RETURN ``Ok"
\ENDIF
\end{algorithmic}

\end{algorithm}

Intuitively, the idea behind the separation oracle is quite simple. Suppose that we want to verify that the variable assignment $\tilde{p}_1', \ldots, \tilde{p}_m',\mathcal{P}_{Adv,B}'$ is feasible. The first few steps of our separation oracle verify that constraints (1)--(4) from Goal \ref{goal:MinAdvSuccessAsLP} are satisfied by the assignment $\tilde{p}_1', \ldots, \tilde{p}_m'$. These straightforward checks simply verify that the proposed CASH distribution $\tilde{p}_1', \ldots, \tilde{p}_m'$ is valid and that the server's amortized costs are less than $C_{SRV,\alpha}$. 

The next step,  verifying that all type (5) constraints are satisfied, is a bit more challenging because there are exponentially many constraints. Recall that these constraints intuitively ensure that $\mathcal{P}_{Adv,B}'$ is indeed an upper bound on the success rate of the optimal adversary given CASH distribution $\tilde{p}_1', \ldots, \tilde{p}_m'$. While we don't have time to check every feasible budget allocation $\vec{b} \in \mathcal{F}_B$ for the adversary, it suffices to find the adversary's optimal budget allocation $\vec{b}'$  and verify that $\mathcal{P}_{Adv,B}'$ is an upper bound on the adversary's success rate given allocation $\vec{b}'$. 

The adversary gets $\lfloor B/k \rfloor$ total guesses of the form $\left(pw_i,j\right)$ for each user $u$. The probability that the guess $\left(pw_i,j\right)$ is correct is simply $p_{i,j}' \doteq p_i\cdot\tilde{p}_j'$  --- the guess is correct if and only if $u$ selected password $pwd_u = pwd_i$ and we selected CASH running time parameter $t_u = j$. The adversary's optimal strategy is simple: try the $\lfloor B/k \rfloor$ most likely guesses. Thus, we can quickly find the adversary's optimal budget allocation $\vec{b}'$ by computing $p_{i,j}$ for each pair $\left(pw_i,j\right)$ and sorting these values. This takes time $O\left( nm \log nm\right)$. 

\begin{remindertheorem}{\ref{thm:SeparationOracle}} \thmSepOracle
\end{remindertheorem}
\vspace{4mm}

\begin{proofof}{Theorem \ref{thm:SeparationOracle}} (Sketch)
The most expensive step in our algorithm is sorting the $p_{i,j}'$ values. There are $mn$ such values so the algorithm takes $O\left(mn \log mn\right)$ steps. We now argue that our separation oracle has correct behavior. 

Suppose first that there is a constraint $C$ from Optimization Goal \ref{goal:MinAdvSuccessAsLP} that is not satisfied by $\tilde{p}_1',\ldots,\tilde{p}_m',\mathcal{P}_{ADV,B}'$. It is easy to verify that our separation oracle will catch violations of constraints (1)--(4) so we can assume that $C$ be a violated type (5) constraint $\mathcal{P}_{ADV,B}' < \sum_{i=1}^n p_i \sum_{j=1}^{b_i^C} \tilde{p}_j$ where $\left(b_1^C,\ldots,b_n^C \right) \in \mathcal{F}_B$. Let $b_1',\ldots, b_n'$ denote the budget obtained by sorting the $p_{i,j}$ values and then greedily selecting a set $S'$ of the largest values until the budget expires --- we define $b_i'$ to be the number of values of the form $p_{i,j}$ that are selected and $S' = \{ \left(i,j \right)~\vline~i\leq n \wedge j \leq b_i'\}$ It suffices to argue that \[\sum_{i=1}^n p_i \sum_{j=1}^{b_i'} \tilde{p}_j \geq \sum_{i=1}^n p_i \sum_{j=1}^{b_i^C} \tilde{p}_j\] because in this case our algorithm will return the violated constraint
\[\mathcal{P}_{ADV,B}' \geq \sum_{i=1}^n p_i \sum_{j=1}^{b_i'} \tilde{p}_j \ . \]
Let $S^C = \{ (i,j)~\vline~ i \leq n \wedge j \leq b_i^C \}$. We first observe that 
\[ \sum_{(i,j) \in S'} p_{i,j}' \geq \sum_{(i,j) \in S^C} p_{i,j}'   \ , \]
by construction of $S'$. Thus, 
\begin{eqnarray*}
\sum_{i=1}^n p_i \sum_{j=1}^{b_i'} \tilde{p}_j' &=& \sum_{(i,j) \in S'} p_{i,j}' \\
&\geq& \sum_{(i,j) \in S^C} \\ 
&=& \sum_{i=1}^n p_i \sum_{j=1}^{b_i^C} \tilde{p}_j'
\end{eqnarray*}

Finally, when the solution $\tilde{p}_1',\ldots,\tilde{p}_m',\mathcal{P}_{ADV,B}'$ does satisfy all constraints from Optimization Goal \ref{goal:MinAdvSuccessAsLP} our algorithm will not find a constraint $b_1',\ldots,b_n'$ such that  
\[\mathcal{P}_{ADV,B}' < \sum_{i=1}^n p_i \sum_{j=1}^{b_i'} \tilde{p}_j  \ . \]
In this case our algorithm will return ``Ok" --- the desired outcome. 
\end{proofof}

\begin{algorithm}[H]
 \caption{$\mathbf{InitialConstraints}\left(C_{max},\alpha,k\right)$}
 \label{alg:InitialConstraints}
\begin{algorithmic}[1]
 \renewcommand{\algorithmicrequire}{\textbf{Input:}}
 \renewcommand{\algorithmicensure}{\textbf{Output:}}
 \REQUIRE  $C_{max},\alpha,k$
 \STATE $C \gets C \bigcup \{\sum_{i=1}^m \tilde{p}_i = 1 \}$.
\STATE $C \gets C \bigcup \{1\geq \tilde{p}_m \geq 0 \}$.
\STATE $C \gets C \bigcup \{(1-\alpha) m\cdot k + \alpha\cdot k \sum_{i=1}^m i\cdot\tilde{p}_i \leq C_{SRV,\alpha} \}$. 
\FOR{i=1,\ldots,m-1}
\STATE $C \gets C \bigcup \{1 \geq \tilde{p}_i \geq 0\}$.
\STATE $C \gets C \bigcup \{\tilde{p}_i \geq \tilde{p}_{i+1}\}$.
\ENDFOR
\STATE $C \gets C \bigcup \{1 \geq \mathcal{P}_{Adv,B} \geq 0 \}$.
\end{algorithmic}
\end{algorithm}

\begin{algorithm}[H]
 \caption{$\mathbf{RationalAdvSuccess}\left(p,n,\hat{v}, \tilde{p}, k \right) $}
 \label{alg:AdvSuccess}
\begin{algorithmic}[1]
 \renewcommand{\algorithmicrequire}{\textbf{Input:}}
 \renewcommand{\algorithmicensure}{\textbf{Output:}}
\REQUIRE $p_1,\ldots, p_{n'}$, $n_1,\ldots,n_{n'}$,  $\hat{v}$, $\tilde{p}$
\STATE $curSuccess \gets 0$
\STATE $curThreshold \gets 0$
\STATE $curUtility \gets 0$
\STATE $bestUtilityFound \gets 0$
\STATE $bestUtilitySuccess \gets 0$
\STATE $T \gets \emptyset$
\FOR{$i=1,\ldots, n'$}
\FOR{$j=1,\ldots, m$}
\STATE $T.Add\left(p_i \cdot \tilde{p}_j, n_i \right)$
\ENDFOR
\ENDFOR
\STATE $\mathbf{Sort}\left(T\right)$.   ~~~~~~\COMMENT{Use first component $p_i\cdot \tilde{p}_j$ for} \STATE~~~~~~~~~~~~~~~~~~\COMMENT{ comparison (greatest to least)}
\FOR{$t \in T$}
\STATE $(\pi, count) \gets t$
\STATE $curThreshold \gets curThreshold + count$
\STATE $curSuccess \gets \pi \cdot count$
\STATE $\Delta benefit \gets \hat{v} \cdot  \pi \cdot count$
\STATE $\Delta cost \gets  k*\left(count*(1-curSuccess)+ \frac{\pi\cdot count^2 + \pi \cdot count}{2}\right)$
\STATE $curUtility \gets curUtility + \Delta benefit - \Delta cost$

\IF{$curUtility  > bestUtilityFound$}
\STATE $bestUtilityFound \gets curUtility $
\STATE $bestUtilitySuccess \gets curSuccess$
\ENDIF
\ENDFOR
\RETURN $bestUtilitySuccess$
\end{algorithmic}

\end{algorithm}

While we do not have a polynomial time algorithm to compute the Stackelberg equilibrium of our game, it is always easy for the adversary to compute his best response. 

\begin{theorem}
Let $p = p_1 \geq \ldots \geq p_{n'}$ and $n_1,\ldots, n_{n'}$ define a probability distribution over passwords in which there are $n_i$ passwords that each are chosen with probability $p_i$ and let $\tilde{p}_1 \geq \ldots \geq \tilde{p}_m$ denote any CASH distribution. Then for any value $\hat{v}$ and any hash cost parameter $k$ we can computed the adversary's optimal strategy in time $O(mn' \log mn')$.
\end{theorem}
\begin{proof} (sketch)
Algorithm \ref{alg:AdvSuccess} computes the adversay's optimal strategy. The most expensive step is the sorting the $mn$ tuples, which takes time $O(mn' \log mn')$. Thus, Algorithm \ref{alg:AdvSuccess} runs in time in time $O(mn' \log mn')$. Algorithm \ref{alg:AdvSuccess} iterates through the different possible thresholds that a rational adversary might select. The variable $curUtility$ keeps track of the utility at each threshold allowing us to remember which threshold was optimal. Intuitively, Algorithm \ref{alg:AdvSuccess} will find the best strategy if and only if $curUtility$ is always a correct estimate of the adversary's utility. Clearly, this is true initially (the utility of selecting $B^*=0$ is $0$). Thus, by induction, it suffices to show that the formulas used to compute marginal cost and marginal benefit are correct. If the adversary adds all of the tuples $(pwd,t)$ corresponding to $(\pi,count)$ to the set of tuples to guess then the adversary is increasing the odds that he cracks the password by $\pi\cdot count$ because he is adding $count$ tuples to his set of guesses and each tuple is correct with probability $\pi$. Thus, his marginal benefit is $\hat{v}\cdot\pi \cdot count$. To analyze marginal cost we consider three cases:
1) The correct tuple $(pwd*,t*)$ was already in the adversary's set of tuples to guess. In this case we don't increase the adversary's guessing costs because he will always quit before he guesses one of the new tuples we added. 2) The correct tuple $(pwd*,t*)$ is not already in the adversary's set of tuples and it is not in the new set of tuples we add. Thus, we increase the adversary's guessing costs  by $k*count$. 3) The correct tuple $(pwd*,t*)$ is  in the new set of tuples we add. In this case we increase the adversary's guessing costs by $k*\left(\frac{count+1}{2} \right)$ in expectation. The probability that we are in case 2 is $(1-curSuccess)$ and the probability that we are in case 3 is $\pi\cdot count$. Thus, 
  \[\Delta cost \gets  k*\left(count*(1-curSuccess)+ \frac{\pi\cdot count^2 + \pi \cdot count}{2}\right) \ .\]
\end{proof}

\section*{Extra Plots}
Figures \ref{fig:YahooResultsOptForRockYou95} and \ref{fig:RockYouResultsOptForYahoo95} explore what happens when the defender uses the wrong empirical password distribution when searching for a good CASH distribution $\tilde{p}$ (e.g., if the defender optimizes $\tilde{p}$ under the assumption that the empirical password distribution is given by the Yahoo! dataset when the actual distribution is given by the RockYou dataset).  Once again non-uniform CASH and CASH both significantly outperform deterministic key-stretching, an non-uniform CASH outperforms uniform CASH (slightly) over most of the curve\footnote{The exception is Figure \ref{fig:YahooResultsOptForRockYou95} contains a region where uniform-CASH actually outperforms non-uniform CASH (yielding $15\%$ reduction in cracked passwords in comparison to non-uniform CASH).}. Interestingly, in one part of the curve in Figure \ref{fig:YahooResultsOptForRockYou95} the adversary's success rate actually drops as $v$ increases. This would be impossible if the defender was using the correct empirical password distribution. In this case the adversary's success rate drops when $v$ increases because the defender switches to a better CASH distribution $\tilde{p}$ that happens to perform better under the real distribution. 

Figure \ref{fig:YahooRobust3Results95} plots the fraction of cracked passwords against a value $v$ adversary when the defender selects $\tilde{p}$ and $k$ under the assumption that the adversary's value is $\hat{v}=2.9\times C_{max} \times 10^7$ (using the empirical password distribution from the Yahoo dataset and setting $\alpha=0.95$).
Figure \ref{fig:YahooDistribution95} plots the corresponding cumulative cost distribution for the authentication server induced by $\tilde{p}$, $k$ and $\alpha$. For comparison, we also include the cumulative cost distributions for uniform CASH and deterministic key-stretching under the same maximum cost parameter $C_{max}$. 

\begin{figure}[!t]
\centering
    \begin{tikzpicture}[scale=1.3]  
   \begin{semilogxaxis}[
    title style={align=center},
    title={ {\small $\alpha=0.95$ }},
    xlabel={Adversary Value: $\frac{v}{C_{max}}$},
    xlabel={Adversary Value: $\frac{v}{C_{max}}$},
    ylabel={$\%$ Cracked Passwords},
    ylabel shift = -3pt,
    grid=major,
    small,
    cycle list = {{red, mark=none}, {blue, mark=none}, {green, mark=none},{green, dashed, mark=none}, {blue, dashed, mark=none},  {blue, dashed, mark=none}, {red, dashed, mark=none},{brown, dashed, mark=none} },
    legend style = {font=\tiny, at={(.05,.95)}, anchor=north west},
    legend entries = { $\PAdvCASH{v}{v}{C_{max}}$, $\PAdvUnif{v}{C_{max}}$, $\PAdvDet{v}{C_{max}}$, Improvement: \textcolor{blue}{blue}$-$\textcolor{red}{red}, Improvement: \textcolor{green}{green}$-$\textcolor{red}{red}}
   ] 

    \addlegendimage{no markers, red, mark=square*}
    \addlegendimage{no markers, blue, mark=square*}
    \addlegendimage{no markers, green, mark=square*}
    \addlegendimage{no markers, dashed, blue}
    \addlegendimage{no markers, dashed, green}
\addplot coordinates {  (100,0)  (500,0.00799779034835097)  (1000,0.0125851579477658)  (5000,0.0275522434975735)  (10000,0.0465878576790846)  (50000,0.117449398351743)  (100000,0.155467102277563)  (500000,0.256128351889682)  (1000000,0.310620114773622)  (5000000,0.452265658830308)  (10000000,0.535874002906533)  (15000000,1)  (20000000,0.999999999999985)  (25000000,0.999999999999995)  (26500000,1.00000000000001)  (27000000,1)  (27500000,0.999999999999993)  (28000000,1.00000000000002)  (29000000,0.999999999999994)  (30000000,0.999999999999994)  (70000000,1.00000000000001)  (100000000,1.00000000000001)  };
 \addplot coordinates {  (100,0)  (500,0.00891714075850031)  (1000,0.0136977788934083)  (5000,0.0265629142590948)  (10000,0.0465878576790853)  (50000,0.124664436714369)  (100000,0.164009427486493)  (500000,0.264563609156202)  (1000000,0.322603129466177)  (5000000,0.465731506185799)  (10000000,0.55428368364662)  (15000000,1)  (20000000,1)  (25000000,1)  (26500000,1)  (27000000,1)  (27500000,1)  (28000000,1)  (29000000,1)  (30000000,1)  (70000000,1)  (100000000,1)  };
 \addplot coordinates {  (100,0)  (500,0.0136977788934083)  (1000,0.0180747779954648)  (5000,0.0441965417827129)  (10000,0.0730608733055595)  (50000,0.160089926850547)  (100000,0.199287847017616)  (500000,0.317302545367371)  (1000000,0.376634078642379)  (5000000,0.518411276766697)  (10000000,1)  (15000000,1)  (20000000,1)  (25000000,1)  (26500000,1)  (27000000,1)  (27500000,1)  (28000000,1)  (29000000,1)  (30000000,1)  (70000000,1)  (100000000,1)  };
\addplot coordinates {  (100,0)  (500,0.00569998854505728)  (1000,0.00548962004769893)  (5000,0.0166442982851393)  (10000,0.0264730156264749)  (50000,0.0426405284988036)  (100000,0.0438207447400539)  (500000,0.0611741934776891)  (1000000,0.0660139638687569)  (5000000,0.0661456179363884)  (10000000,0.464125997093467)  (15000000,-1.99840144432528E-15)  (20000000,1.45439216225896E-14)  (25000000,4.66293670342566E-15)  (26500000,-6.66133814775094E-15)  (27000000,-1.99840144432528E-15)  (27500000,7.105427357601E-15)  (28000000,-1.55431223447522E-14)  (29000000,5.6621374255883E-15)  (30000000,5.6621374255883E-15)  (70000000,-6.66133814775094E-15)  (100000000,-6.66133814775094E-15)  };
\addplot coordinates {  (100,0)  (500,0.000919350410149339)  (1000,0.00111262094564242)  (5000,-0.000989329238478758)  (10000,7.70217223333702E-16)  (50000,0.00721503836262524)  (100000,0.00854232520893045)  (500000,0.00843525726651972)  (1000000,0.0119830146925548)  (5000000,0.0134658473554905)  (10000000,0.0184096807400869)  (15000000,-1.99840144432528E-15)  (20000000,1.45439216225896E-14)  (25000000,4.66293670342566E-15)  (26500000,-6.66133814775094E-15)  (27000000,-1.99840144432528E-15)  (27500000,7.105427357601E-15)  (28000000,-1.55431223447522E-14)  (29000000,5.6621374255883E-15)  (30000000,5.6621374255883E-15)  (70000000,-6.66133814775094E-15)  (100000000,-6.66133814775094E-15)  };

   \end{semilogxaxis} 

  \end{tikzpicture}

\caption{RockYou Results (Optimized for Yahoo): $\alpha = 0.95$.}
\label{fig:RockYouResultsOptForYahoo95}
\end{figure}
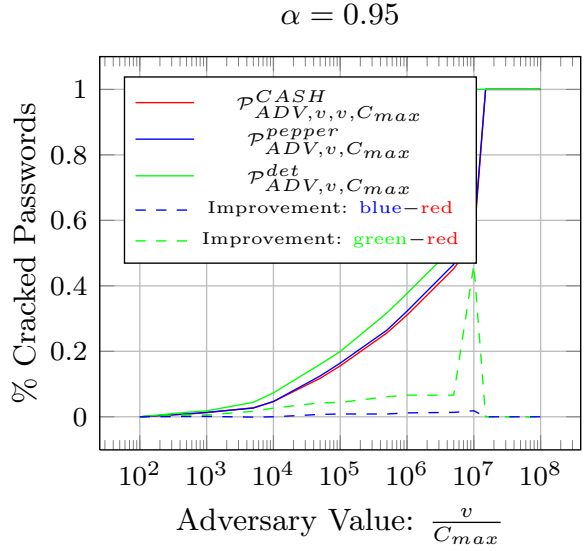

\begin{figure}[!t]
\centering
    \begin{tikzpicture}[scale=1.3]  
   \begin{semilogxaxis}[
    title style={align=center},
    title={ {\small $\alpha=0.95$ }},
     xlabel={Adversary Value: $\frac{v}{C_{max}}$},
    ylabel={$\%$ Cracked Passwords},
    ylabel shift = -3pt,
    grid=major,
    small,
    cycle list = {{red, mark=none}, {blue, mark=none}, {green, mark=none},{green, dashed, mark=none}, {blue, dashed, mark=none},  {blue, dashed, mark=none}, {red, dashed, mark=none},{brown, dashed, mark=none} },
    legend style = {font=\tiny, at={(.05,.95)}, anchor=north west},
    legend entries = { $\PAdvCASH{v}{v}{C_{max}}$, $\PAdvUnif{v}{C_{max}}$, $\PAdvDet{v}{C_{max}}$, Improvement: \textcolor{blue}{blue}$-$\textcolor{red}{red}, Improvement: \textcolor{green}{green}$-$\textcolor{red}{red}}
   ] 

    \addlegendimage{no markers, red, mark=square*}
    \addlegendimage{no markers, blue, mark=square*}
    \addlegendimage{no markers, green, mark=square*}
    \addlegendimage{no markers, dashed, blue}
    \addlegendimage{no markers, dashed, green}
\addplot coordinates {  (100,0)  (500,0.00857947813136685)  (1000,0.0130192581998815)  (5000,0.0196877875530742)  (10000,0.0290163537920214)  (50000,0.0724627397359508)  (100000,0.0997684831358242)  (500000,0.175899274729135)  (1000000,0.223577046597311)  (5000000,0.351703707749315)  (10000000,0.415825677680902)  (15000000,0.615363602360774)  (20000000,0.466773219038238)  (25000000,0.501655087428729)  (26500000,0.504104132587705)  (27000000,0.5057774892863)  (27500000,0.5057774892863)  (28000000,0.99999999999999)  (29000000,0.99999999999999)  (30000000,0.99999999999999)  (70000000,0.99999999999999)  (100000000,0.99999999999999)  };
 \addplot coordinates {  (100,0)  (500,0.0108687224894377)  (1000,0.0130192581998815)  (5000,0.0196877875530742)  (10000,0.0277625668318636)  (50000,0.0756145584896869)  (100000,0.104286400708258)  (500000,0.181055900840701)  (1000000,0.230322497241287)  (5000000,0.364249437207828)  (10000000,0.420317345392629)  (15000000,0.45834913690049)  (20000000,0.478853387778074)  (25000000,0.50975028086399)  (26500000,0.50975028086399)  (27000000,0.50975028086399)  (27500000,0.50975028086399)  (28000000,1)  (29000000,1)  (30000000,1)  (70000000,1)  (100000000,1)  };
 \addplot coordinates {  (100,0.0108687224894377)  (500,0.0130192581998815)  (1000,0.0130192581998815)  (5000,0.0273838584095427)  (10000,0.0441510096695537)  (50000,0.101412040578669)  (100000,0.132372669808665)  (500000,0.224601121331902)  (1000000,0.282624504055384)  (5000000,0.411368138539665)  (10000000,0.478853387778074)  (15000000,1)  (20000000,1)  (25000000,1)  (26500000,1)  (27000000,1)  (27500000,1)  (28000000,1)  (29000000,1)  (30000000,1)  (70000000,1)  (100000000,1)  };
\addplot coordinates {  (100,0.0108687224894377)  (500,0.00443978006851463)  (1000,1.38777878078145E-17)  (5000,0.00769607085646846)  (10000,0.0151346558775323)  (50000,0.0289493008427179)  (100000,0.0326041866728405)  (500000,0.0487018466027666)  (1000000,0.0590474574580722)  (5000000,0.0596644307903499)  (10000000,0.0630277100971723)  (15000000,0.384636397639226)  (20000000,0.533226780961762)  (25000000,0.498344912571271)  (26500000,0.495895867412295)  (27000000,0.4942225107137)  (27500000,0.4942225107137)  (28000000,9.65894031423886E-15)  (29000000,9.65894031423886E-15)  (30000000,9.65894031423886E-15)  (70000000,9.65894031423886E-15)  (100000000,9.65894031423886E-15)  };
\addplot coordinates {  (100,0)  (500,0.00228924435807083)  (1000,1.38777878078145E-17)  (5000,-3.46944695195361E-18)  (10000,-0.00125378696015783)  (50000,0.00315181875373605)  (100000,0.00451791757243336)  (500000,0.00515662611156589)  (1000000,0.00674545064397564)  (5000000,0.0125457294585128)  (10000000,0.00449166771172749)  (15000000,-0.157014465460284)  (20000000,0.012080168739836)  (25000000,0.00809519343526022)  (26500000,0.00564614827628496)  (27000000,0.00397279157768993)  (27500000,0.00397279157768993)  (28000000,9.65894031423886E-15)  (29000000,9.65894031423886E-15)  (30000000,9.65894031423886E-15)  (70000000,9.65894031423886E-15)  (100000000,9.65894031423886E-15)  };

   \end{semilogxaxis} 
  \end{tikzpicture}

\caption{Yahoo Results (Optimized for RockYou): $\alpha = 0.95$.}
\label{fig:YahooResultsOptForRockYou95}
\end{figure}
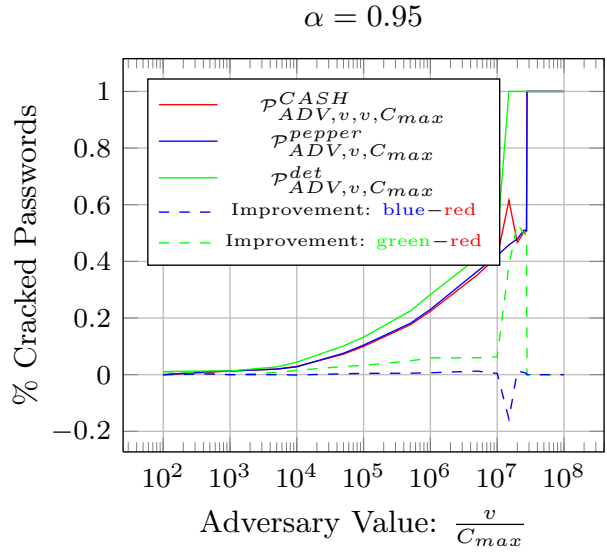

 \begin{figure}[!t]
\centering
    \begin{tikzpicture}[scale=1.3]  
   \begin{semilogxaxis}[
    title style={align=center},
    title={ {\small $\alpha=0.95$, $\hat{v} = 2.9\times C_{max} \times 10^7$ }},
       xlabel={Adversary Value: $\frac{v}{C_{max}}$},
    ylabel shift = -3pt,
    grid=major,
    small,
    cycle list = {{red, mark=none}, {blue, mark=none}, {green, mark=none},{blue, dashed, mark=none}, {green, dashed, mark=none},  {purple, dotted, mark=none}, {red, dashed, mark=none},{brown, dashed, mark=none} },
    legend style = {font=\tiny, at={(.05,.95)}, anchor=north west},
    legend entries = { $\PAdvCASH{v}{\hat{v}}{C_{max}}$, $\PAdvUnif{v}{C_{max}}$, $\PAdvDet{v}{C_{max}}$, Improvement: \textcolor{blue}{blue}$-$\textcolor{red}{red}, Improvement: \textcolor{green}{green}$-$\textcolor{red}{red}, $v = \hat{v}$  }
   ] 

    \addlegendimage{no markers, red, mark=square*}
    \addlegendimage{no markers, blue, mark=square*}
    \addlegendimage{no markers, green, mark=square*}
    \addlegendimage{no markers, dashed, blue}
    \addlegendimage{no markers, dashed, green}
    \addlegendimage{no markers, dotted, purple}

\addplot coordinates {  (100,0.00433975273329383)  (500,0.0142795946346605)  (1000,0.0173591504206237)  (5000,0.0378169183085169)  (10000,0.0494400754971486)  (50000,0.0964486971826689)  (100000,0.128757535130007)  (500000,0.22001064231512)  (1000000,0.261455561816934)  (5000000,0.403458065077915)  (10000000,0.597510122111095)  (15000000,0.620465369472821)  (20000000,0.638899424600312)  (25000000,0.652568925185368)  (26500000,0.652568925185368)  (27000000,0.652568925185368)  (27500000,0.652568925185368)  (28000000,0.652568925185368)  (29000000,0.652568925185368)  (30000000,0.673166853909312)  (70000000,0.999999999999985)  (100000000,0.999999999999985)  };
\addplot coordinates {  (100,0)  (500,0.0108687224894377)  (1000,0.0130192581998815)  (5000,0.0196877875530742)  (10000,0.0277625668318636)  (50000,0.0756145584896869)  (100000,0.104286400708258)  (500000,0.181055900840701)  (1000000,0.230322497241287)  (5000000,0.364249437207828)  (10000000,0.420317345392629)  (15000000,0.45834913690049)  (20000000,0.478853387778074)  (25000000,0.50975028086399)  (26500000,0.50975028086399)  (27000000,0.50975028086399)  (27500000,0.50975028086399)  (28000000,1)  (29000000,1)  (30000000,1)  (70000000,1)  (100000000,1)  };
 \addplot coordinates {  (100,0.0108687224894377)  (500,0.0130192581998815)  (1000,0.0130192581998815)  (5000,0.0273838584095427)  (10000,0.0441510096695537)  (50000,0.101412040578669)  (100000,0.132372669808665)  (500000,0.224601121331902)  (1000000,0.282624504055384)  (5000000,0.411368138539665)  (10000000,0.478853387778074)  (15000000,1)  (20000000,1)  (25000000,1)  (26500000,1)  (27000000,1)  (27500000,1)  (28000000,1)  (29000000,1)  (30000000,1)  (70000000,1)  (100000000,1)  };

\addplot coordinates {  (100,-0.00433975273329383)  (500,-0.0034108721452228)  (1000,-0.00433989222074217)  (5000,-0.0181291307554427)  (10000,-0.021677508665285)  (50000,-0.020834138692982)  (100000,-0.0244711344217496)  (500000,-0.0389547414744192)  (1000000,-0.0311330645756472)  (5000000,-0.0392086278700874)  (10000000,-0.177192776718466)  (15000000,-0.162116232572331)  (20000000,-0.160046036822238)  (25000000,-0.142818644321378)  (26500000,-0.142818644321378)  (27000000,-0.142818644321378)  (27500000,-0.142818644321378)  (28000000,0.347431074814632)  (29000000,0.347431074814632)  (30000000,0.326833146090688)  (70000000,1.46549439250521E-14)  (100000000,1.46549439250521E-14)  };
\addplot coordinates {  (100,0.00652896975614386)  (500,-0.001260336434779)  (1000,-0.00433989222074217)  (5000,-0.0104330598989743)  (10000,-0.00528906582759489)  (50000,0.00496334339599985)  (100000,0.00361513467865751)  (500000,0.00459047901678147)  (1000000,0.0211689422384493)  (5000000,0.00791007346174977)  (10000000,-0.118656734333021)  (15000000,0.379534630527179)  (20000000,0.361100575399688)  (25000000,0.347431074814632)  (26500000,0.347431074814632)  (27000000,0.347431074814632)  (27500000,0.347431074814632)  (28000000,0.347431074814632)  (29000000,0.347431074814632)  (30000000,0.326833146090688)  (70000000,1.46549439250521E-14)  (100000000,1.46549439250521E-14)  };

\addplot coordinates { (29000000, 0)  (29000000, 0.5)  (29000000, 1) };
   \end{semilogxaxis} 
  \end{tikzpicture}
  
  \caption{Yahoo: $\hat{v}\neq v$.}
\label{fig:YahooRobust3Results95}
\end{figure}
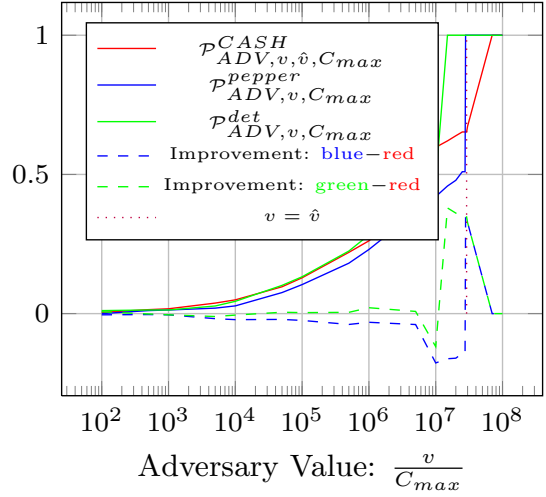

 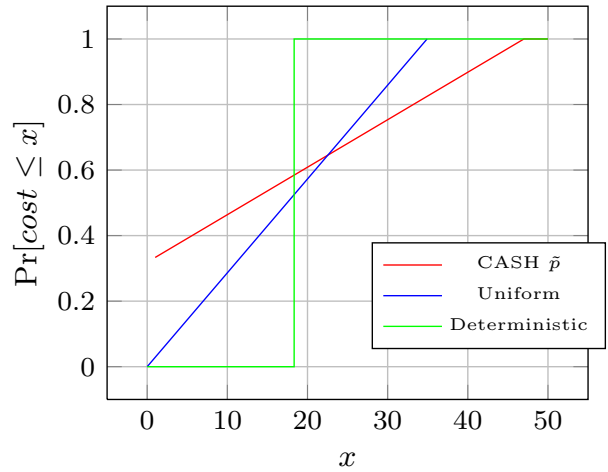
\begin{figure}[!t]
\centering
    \begin{tikzpicture}[scale=1.3]  
   \begin{axis}[
    title style={align=center},
    title={ {\small $\alpha=0.95$, $\hat{v} = 2.9\times C_{max} \times 10^7$ }},
       xlabel={$x$},
    ylabel={$\Pr[cost \leq x]$},
    ylabel shift = -3pt,
    grid=major,
    small,
    cycle list = {{red, mark=none}, {blue, mark=none}, {green, mark=none},{blue, dashed, mark=none}, {green, dashed, mark=none},  {purple, dotted, mark=none}, {red, dashed, mark=none},{brown, dashed, mark=none} },
    legend style = {font=\tiny, at={(.55,.40)}, anchor=north west},
    legend entries = { CASH $\tilde{p}$, Uniform, Deterministic  }
   ] 

    \addlegendimage{no markers, red, mark=square*}
    \addlegendimage{no markers, blue, mark=square*}
    \addlegendimage{no markers, green, mark=square*}
    \addlegendimage{no markers, dashed, blue}
    \addlegendimage{no markers, dashed, green}
    \addlegendimage{no markers, dotted, purple}
\addplot coordinates { (1,0.333333333333333) (2,0.347826086956522) (3,0.36231884057971) (4,0.376811594202899) (5,0.391304347826087) (6,0.405797101449275) (7,0.420289855072464) (8,0.434782608695652) (9,0.449275362318841) (10,0.463768115942029) (11,0.478260869565217) (12,0.492753623188406) (13,0.507246376811594) (14,0.521739130434783) (15,0.536231884057971) (16,0.550724637681159) (17,0.565217391304348) (18,0.579710144927536) (19,0.594202898550724) (20,0.608695652173913) (21,0.623188405797101) (22,0.637681159420289) (23,0.652173913043478) (24,0.666666666666666) (25,0.681159420289855) (26,0.695652173913043) (27,0.710144927536231) (28,0.72463768115942) (29,0.739130434782608) (30,0.753623188405796) (31,0.768115942028985) (32,0.782608695652173) (33,0.797101449275361) (34,0.81159420289855) (35,0.826086956521738) (36,0.840579710144927) (37,0.855072463768115) (38,0.869565217391303) (39,0.884057971014492) (40,0.89855072463768) (41,0.913043478260868) (42,0.927536231884057) (43,0.942028985507245) (44,0.956521739130433) (45,0.971014492753622) (46,0.98550724637681) (47,0.999999999999998) (48,0.999999999999998) (49,0.999999999999998) (50,0.999999999999998) };
\addplot coordinates { (0,0)  (34.9174757204507,1) };
\addplot coordinates { (0,0) (18.3333333333333,0) (18.3333333333333,1) (50,1) };

   \end{axis} 
  \end{tikzpicture}
  
  \caption{Yahoo: CASH Cumulative Probability Distribution.}
\label{fig:YahooDistribution95}
\end{figure}

\end{document}